\documentclass[10pt,a4paper]{article}

\usepackage{amsmath,amsgen,latexsym}
\usepackage{amstext,amssymb,amsfonts,latexsym}
\usepackage{theorem}
\usepackage{pifont}
\usepackage[dvips]{graphics,epsfig}
\usepackage[dvips]{graphicx}

 \newcommand{\bs}{\bigskip}
 \newcommand{\ms}{\medskip}
 \newcommand{\n}{\noindent}
 \newcommand{\s}{\smallskip}
 \newcommand{\hs}[1]{\hspace*{ #1 mm}}
 \newcommand{\vs}[1]{\vspace*{ #1 mm}}

 \newcommand{\setempty}{\mathrm{\O}}
 \newcommand{\real}{\mathbb{R}}

 \newcommand{\nat}{\mathbb{N}}
 
 \newcommand{\integer}{\mathbb{Z}}
 \newcommand{\rational}{\mathbb{Q}}
 
 \newcommand{\complex}{\mathbb{C}}

 \newcommand{\algebraic}{\mathbb{A}}

 \newcommand{\prob}{{\mathrm{Prob}}}

 \newcommand{\co}{\mathrm{co}\mbox{-}}

 \newcommand{\ie}{\textrm{i.e.},\hspace*{2mm}}
 \newcommand{\eg}{\textrm{e.g.},\hspace*{2mm}}
 
 \newcommand{\etalc}{\textrm{et al.}}

 \newcommand{\CC}{{\cal C}}

 \newcommand{\cequalp}{\mathrm{C}_{=}\mathrm{P}}

 \newcommand{\reg}{\mathrm{REG}}

 \newcommand{\lmatrices}[4]{\left[ \begin{array}{cc} #1 & #2 \\%
      #3 & #4   \end{array}\right]}

 \newcommand{\ninematrices}[9]{\left[ \begin{array}{ccc} #1 & #2 & #3 \\%
      #4 & #5 & #6 \\  #7 & #8 & #9   \end{array}\right]}

 \newcommand{\twothreematrices}[6]{\left[ \begin{array}{ccc} #1 & #2 & #3 \\%
       #4 & #5 & #6   \end{array}\right]}

\theoremstyle{plain}
\theoremheaderfont{\bfseries}
 \newtheorem{theorem}{Theorem}[section]
 \newtheorem{lemma}[theorem]{Lemma}
 \newtheorem{proposition}[theorem]{Proposition}
 \newtheorem{corollary}[theorem]{Corollary}

 {\theorembodyfont{\rmfamily}
  }
 {\theorembodyfont{\rmfamily} }
 {\theorembodyfont{\rmfamily} }

 \newtheorem{claim}{Claim}

 \newenvironment{proof}{\par \noindent
            {\bf Proof. \hs{2}}}{\hfill$\Box$ \vspace*{3mm}}

 \newenvironment{proofsketch}{\par \noindent
            {\bf Proof Sketch. \hs{2}}}{\hfill$\Box$ \vspace*{3mm}}

 \newenvironment{proofof}[1]{\vspace*{5mm} \par \noindent
         {\bf Proof of #1.\hs{2}}}{\hfill$\Box$ \vspace*{3mm}}

 \newcommand{\ceilings}[1]{\lceil #1 \rceil}
 
 \newcommand{\pair}[1]{\langle #1 \rangle}

 \newcommand{\qubit}[1]{| #1 \rangle}
 \newcommand{\bra}[1]{\langle #1 |}
 \newcommand{\ket}[1]{| #1 \rangle}

\newcommand{\ignore}[1]{}

\newcommand{\cent}{{|}\!\!\mathrm{c}}
\newcommand{\dollar}{\$}

 \newcommand{\am}{\mathrm{AM}}

 \newcommand{\twopfa}{\mathrm{2PFA}}

 \newcommand{\twoppfa}{\mathrm{2PPFA}}

 \newcommand{\tworfa}{\mathrm{2RFA}}

 \newcommand{\twoqfa}{\mathrm{2QFA}}
 \newcommand{\twobqfa}{\mathrm{2BQFA}}

 \newcommand{\onepqfa}{\mathrm{1PQFA}}
 \newcommand{\onecequalqfa}{\mathrm{1C_{=}QFA}}
 \newcommand{\onenqfa}{\mathrm{1NQFA}}
 \newcommand{\twopqfa}{\mathrm{2PQFA}}
 \newcommand{\twocequalqfa}{\mathrm{2C_{=}QFA}}
 \newcommand{\twoeqfa}{\mathrm{2EQFA}}
 \newcommand{\tworqfa}{\mathrm{2RQFA}}
 \newcommand{\twonqfa}{\mathrm{2NQFA}}

 \newcommand{\twocequalpfa}{\mathrm{2C_{=}PFA}}

 \newcommand{\polytime}{poly\mbox{-}time}

 \newcommand{\lintime}{lin\mbox{-}time}

 \newcommand{\abshalt}{abs\mbox{-}halt}
 \newcommand{\comphalt}{comp\mbox{-}halt}

\begin{document}
\pagestyle{plain}
\setcounter{page}{1}

\begin{center}
{\Large {\bf Complexity Bounds of Constant-Space Quantum Computation}}\footnote{An extended abstract appeared in the Proceedings of the 19th International Conference on Developments in Language Theory (DLT 2015), Liverpool, United Kingdom, July 27--30, 2015, Lecture Notes in Computer Science, Springer, vol.9168, pp.426--438, 2015.} \bs\ms\\

{\sc Tomoyuki Yamakami}\footnote{Present Affiliation: Faculty of Engineering, University of Fukui, 3-9-1 Bunkyo, Fukui 910-8507,  Japan} \bs\\
\end{center}

\begin{quote}
\n{\bf Abstract:}
We realize constant-space quantum computation by measure-many two-way quantum finite automata and evaluate their language recognition power by analyzing patterns of their exotic behaviors and by exploring their structural properties. In particular, we show that, when the automata halt ``in finite steps'' along all computation paths, they must terminate in worst-case liner time. In the bounded-error probability case, the acceptance of the automata depends only on the computation paths that terminate within exponentially many steps even if not all computation paths may terminate.
We also present a classical simulation of those automata on two-way multi-head probabilistic finite automata with cut points.
Moreover, we discuss how the recognition power of the automata varies as the automata's acceptance criteria change to error free, one-sided error,  bounded error, and unbounded error by comparing the complexity of their computational powers.
We further note that, with the use of arbitrary complex transition amplitudes,   two-way unbounded-error quantum finite automata and two-way bounded-error $2$-head quantum finite automata can recognize certain non-recursive languages,  whereas two-way error-free quantum finite automata recognize only recursive languages.

\s

\n{\bf Keywords:} constant space, quantum finite automata, cut point, error free, one-sided error, bounded error, unbounded error, absolutely halt, completely halt, determinant
\end{quote}

\section{Quick Overview}\label{sec:overview}

Computer scientists have primarily concerned themselves with automated mechanical procedures of solving real-life problems in the most practical fashion.  For such practicality, we have paid more attention to ``resources'' used up to execute desired protocols on given computing devices. In order to  build small-scale computing devices, in particular, we are keen to \emph{memory space}, which  stores information or data necessary to carry out a carefully designed protocol on these devices, rather than its running time.
We are particularly interested in devices that consume only a constant amount of memory space,  independent of input size.
Among those devices, we are focused on quantum-mechanical computing devices as a part of the leading Nature-inspired computing paradigm.
Since its introduction in early 1980s, quantum computation theory founded on those devices has significantly evolved. In retrospective, since quantum mechanics is believed by many to govern Nature, it seems inevitable for scientists to have come to inventing  quantum-mechanical computing device.
In quantum computing, when algorithmic procedures require only constant memory space on devices, we have customarily viewed  such devices as \emph{quantum finite automata} (or qfa's), which are a quantum-mechanical extension of classical finite(-state) automata, mainly because they are still capable of storing a fixed amount of useful information by way of manipulating a few number of ``inner states''  even without equipping an additional memory tape. A qfa proceeds its computation simply by applying a finite-dimensional unitary transition matrix and  a set of projective measurements to a linear combination  of qfa's inner states as well as tape head positions. Such simple framework of qfa's is ideal for us to conduct a deeper analysis on the execution of their algorithmic procedures.
Among a variety of qfa models proposed recently (e.g., \cite{ABG+06,Hir10,Nay99}), we are focused mostly on \emph{measure-many two-way quantum finite automata} (or \emph{2qfa's}, for brevity) of Kondacs and Watrous \cite{KW97} because of the simplicity of their definition and the consistency with the past literature \cite{NY04b,NY09,NY15,VY14,Yam14}. Such a model may remain as a core model for the better understandings of fundamental properties of quantum-mechanical constant-memory devices.

In accordance with quantum mechanics, a computation of a 2qfa gradually evolves by applying a unitary transition matrix
to a superposition of \emph{configurations} in a finite-dimensional Hilbert space (called a \emph{configuration space}). Unlike a qfa model of Moore and Crutchfield \cite{MC00}, Kondacs and Watrous's model further uses an operation of observing halting inner states at every computation step. It turns out that allowing its tape head to  move in all directions enables the 2qfa's to attain a significant increase of computational power over 2-way deterministic finite automata, whereas one-way qfa's fail to capture even regular languages  \cite{KW97}. Despite our efforts over the past 20 years, the behaviors of 2qfa's have remained largely enigmatic to us
and the 2qfa's seem to be still awaiting for full investigation of their functionalities.

There are four important issues that we wish to address in depth.

\s

(1) \emph{Acceptance criteria issue.}
The first issue to contemplate is that, in traditional automata theory, recognizing languages by probabilistic finite automata (or pfa's) has been subject to a threshold of the acceptance probability of the automata under the term of ``cut point'' and ``isolated cut point.'' In quantum automata theory, on the contrary, the recognition of languages is originally defined in terms of ``bounded-error probability'' of qfa's \cite{KW97,MC00} although the ``isolated cut point'' criterion has been occasionally used in certain literature (e.g., \cite{BC01}).
What is a precise relationship between those two criteria?
When automata are particularly limited to one-way head moves, as noted in Lemma \ref{PQFA-equal-SL}, the cut-point criterion of pfa's coincides with the unbounded-error criterion of qfa's; however, the same equivalence does not hold in a general 2-way case. In this paper, we shorthandedly denote by $\twobqfa$ the family of languages recognized by bounded-error 2qfa's. When we modify this bounded-error criterion of 2qfa's
to error free (or exact), one-sided error, and unbounded error probabilities,
we further obtain crucial language families\footnote{These notations are analogous to $\mathrm{EQP}$, $\mathrm{RP}$, and $\mathrm{PQP}$ in computational complexity theory.} $\twoeqfa$, $\tworqfa$, and $\twopqfa$, respectively. With the use of `cut point,'' in contrast, two families\footnote{These are associated with $\mathrm{NQP}$ and $\mathrm{C_{=}P}$.} $\twonqfa$ and $\twocequalqfa$ can be thought respectively in terms of zero cut point and nonnegative exact cut points.
In Section \ref{sec:absolute-halt}, we shall state basic relationships among those language families by presenting various inclusions and collapses of them.

(2) \emph{Termination issue.}
One-way qfa's run within $O(n)$ steps but 2-way qfa's are not guaranteed to have running time-bounds. Primarily, similar to 2pfa's, we have been interested in only 2qfa's whose computation paths eventually terminate with at least the 50\% chance. To an arbitrary 2qfa, we cannot implement any internal clock so that the 2qfa terminates its computation at any specified time.
In the past literature, on the contrary, space-bounded quantum computation on quantum Turing machines has been discussed mostly in an extreme case of \emph{absolute halting} (i.e.,
eventual termination of all computation paths) \cite{Wat03}.
Bounded-error 2qfa's that halt absolutely induce
a language family, which is denoted by $\twobqfa(\abshalt)$; in contrast, 2qfa's whose computation paths terminate with probability $1$ (i.e., the probability of non-halting computation is $0$) are said to \emph{halt completely} and introduce another language family $\twobqfa(\comphalt)$. What is the computational power of those language families?
In Section \ref{sec:runtime-2qfa}, we shall observe that, when a 2qfa makes bounded errors, most computation paths of the 2qfa actually terminate in  exponentially many steps. A key to the proof of this phenomenon is the \emph{Dimension Lemma} of Yao \cite{Yao98}, presented in Section \ref{sec:absolute-halt}, which is a direct consequence of an analysis of 2qfa's transition matrices. Furthermore,  when 2qfa's halt absolutely,
we can upper-bound by $O(n)$ the ``worst case'' running time (i.e., the time required for  the longest computation path to terminate) of those 2qfa's, where $n$ refers to input length.

(3) \emph{Transition Amplitude issue.}
Opposed to two-way probabilistic finite automata (or 2pfa's), 2qfa's make their  next moves with certain \emph{(transition) amplitudes}, which are in general  arbitrary complex numbers of absolute values at most $1$.
Since it is possible to encode a large amount of classical information into a few qubits, an early study of  Adleman, DeMarrais, and Huang  \cite{ADH97} on polynomial-time quantum computation revealed that the restriction on the choice of such amplitudes greatly alter the computational power of underlying quantum machines. We thus need to specify a set $K$ of (transition) amplitudes to be used by 2qfa's. For notational convenience, we write $\twoeqfa_{K}$ when all underlying 2qfa's use only amplitudes drawn from $K$.
In Section \ref{sec:non-recursive}, we shall show that $\twoeqfa_{\complex}$ is strictly contained in the family of recursive languages although underlying 2qfa's can manipulate non-recursive amplitudes, where $\complex$ is the set of complex numbers.
In Section \ref{sec:absolute-halt}, we shall prove that $\twoeqfa_{\complex}(\abshalt)$ coincides with $\twoeqfa(\abshalt)$ restricted  to real algebraic amplitudes.
In Section \ref{sec:multi-head-FA}, we shall claim that $\twopqfa_{\rational}$ (where $\rational$ is the set of rational numbers) is different from $\twopqfa_{\complex}$.
As for bounded-error qfa's, however, we shall only say that a multi-head extension of 2qfa's can recognize non-recursive languages. This contrasts the polynomial-time case, in which the language family $\mathrm{BQP}_{\complex}$ (bounded-error quantum polynomial-time) with $\complex$-amplitudes contains non-recursive languages \cite{ADH97}.

(4) \emph{Classical simulation issue.}
Watrous \cite{Wat03} presented a general procedure of simulating space-bounded unbounded-error quantum Turing machines on classical Turing machines with reasonable overhead.
As noted in \cite{NY09}, this simulation leads to the containment  $\twobqfa_{\algebraic}\subseteq\mathrm{PL}\subseteq \mathrm{P}$, where the subscript $\algebraic$ indicates the use of complex algebraic amplitudes and $\mathrm{PL}$ is the  family of all languages recognized by unbounded-error probabilistic Turing machines with $\{0,1/2,1\}$-transition probabilities using $O(\log n)$ space.
In Section \ref{sec:classical-simulation}, we shall give a better complexity upper bound to $\twopqfa$ (and therefore $\twobqfa$) using \emph{multi-head} 2-way probabilistic finite automata with cut points. For this purpose, we shall make an appropriate implementation of a GapL-algorithm of \cite[Theorem 4]{MV97} that computes integer determinants. Notice that, for our implementation, we need to make various changes to the original algorithm. Such changes are necessary because a target matrix is an integer matrix and is given as ``input'' in \cite{MV97}; however, in our case, our target matrix is a real matrix and we need to realize all entries of this matrix in terms of ``acceptance probabilities.''
\ms


The last section will present a short list of challenging questions associated with the aforementioned four issues. Since qfa's may be viewed as a manifestation of quantum mechanics, a deep understanding of the qfa's naturally promotes a better understanding of quantum mechanics in the end. We thus strongly hope that this work stimulates more intensive research activities on the behaviors of qfa's, leading to surprising properties of the qfa's.

\vs{-2}
\sloppy
\section{Basic Notions and Notation}

We quickly review the basic notions and notation necessary to read through the rest of this paper.

\subsection{General Definitions}

Let $\nat$ be the set of all \emph{natural numbers} (that is, nonnegative integers) and set $\nat^{+}=\nat - \{0\}$. Moreover, let $\integer$, $\rational$, $\real$, and $\complex$ denote respectively the sets of all \emph{integers}, of all \emph{rational numbers}, of all \emph{real numbers}, and of all \emph{complex numbers}. The notation $\algebraic$ stands for the set of all \emph{algebraic complex numbers}. For brevity, we write $\imath$ for $\sqrt{-1}$.
Given any complex number $\alpha$, $\alpha^*$ denotes its \emph{conjugate}. The real unit interval between $0$ and $1$ is denoted $[0,1]$. For any two numbers $m,n\in\integer$ with $m\leq n$,
$[m,n]_{\integer}$ expresses the {\em integer interval} between $m$ and $n$; that is, the set $\{m,m+1,m+2,\ldots,n\}$. For brevity, we write $[n]$ for $[1,n]_{\integer}$ for each number $n\in\nat^{+}$.
For any finite set $Q$, $|Q|$ denotes the {\em cardinality} of $Q$.
All vectors in $\complex^n$ are expressed as \emph{column vectors} unless otherwise stated.
Given a number $n\in\nat^{+}$, $M_n(\complex)$ stands for the set of all $n\times n$ complex matrices.
For such a matrix $A=[a_{ij}]\in M_n(\complex)$ and any pair $i,j\in[n]$, the notation $A[i,j]$ refers to $A$'s entry specified by row $i$ and column $j$ of $A$, and  $A_{i,j}$ denotes the submatrix obtained from $A$ by deleting row $i$ and column $j$. The {\em determinant} of $A$ is   $det(A)=\sum_{\sigma}\prod_{i=1}^{n}(-1)^{i+\sigma(i)}a_{i,\sigma(i)}$, where $\sigma$ is taken over all permutations on $[n]$.
The {\em (classical) adjoint} (or {\em adjugate}) of $A$, denoted $adj(A)=[b_{ij}]$, is the matrix whose $(i,j)$-entry is defined by $b_{ij}= (-1)^{i+j}det(A_{i,j})$ for any $i,j\in[n]$.
Assuming that $A$ is nonsingular, it holds that  $A^{-1} = adj(A)/det(A)$; in particular, $A^{-1}[i,j] = adj(A)[i,j]/det(A) = (-1)^{i+j}det(A_{i,j})/det(A)$.
For any complex matrix $A$, the notation $A^T$ and $A^{\dagger}$ respectively denote the \emph{transpose} and the \emph{Hermitian adjoint} of $A$.
For any vector $x$, $\|x\|$ denotes the \emph{$\ell_2$-norm} of $x$ (\ie $\|x\| = (\sum_{i=1}^{n}|x_i|^2)^{1/2}$ if $x=(x_1,x_2,\ldots,x_n)$).
Let $\|A\|$ be the \emph{operator norm} (or \emph{matrix norm})  defined as $\|A\|=\max\{\|Ax\|: \|x\|\neq0\}$ and let $\|A\|_{2}$ be the \emph{Frobenius norm}  $\left(\sum_{i,j}|a_{i,j}|^2\right)^{1/2}$ of $A=[a_{ij}]$.
The {\em trace norm} $\|A\|_{tr}$ of $A$ is $\min\{|Tr(AX^{\dagger})|: \|X\|\leq 1\}$, where $Tr$ indicates the \emph{trace operator}. An important fact is that, if the matrix norm $\|A\|$ is less than $1$, $I-A$ is invertible and $(I-A)^{-1}$ coincides with $\sum_{k=0}^{\infty}A^k$. See, \eg \cite{HJ85} for basic properties of matrices in $M_n(\complex)$.

In general, we use the notation $\Sigma$ for an arbitrary nonempty input alphabet (not necessarily limited to $\{0,1\}$). A {\em string $x$ over $\Sigma$} is a finite sequence of symbols in $\Sigma$ and its {\em length} $|x|$ indicates the number of occurrences of symbols in $x$. In particular, the string of length $0$ is called the {\em empty string} and denoted by $\lambda$. For each number $n\in\nat$, $\Sigma^n$ denotes the set of all strings over $\Sigma$ of length exactly $n$. We write $\Sigma^*$ for $\bigcup_{n\in\nat}\Sigma^n$.
A {\em partial problem} over alphabet $\Sigma$ is a pair $(A,B)$ such that $A,B\subseteq \Sigma^*$ and $A\cap B=\setempty$. When $A\cup B=\Sigma^*$ holds, $B$ becomes the \emph{complement} of $A$ (denoted $\Sigma^*-A$ or simply $\overline{A}$ if $\Sigma$ is clear from the context). We identify $(A,\overline{A})$ with $A$, which is simply called a {\em language}. For a family $\CC$ of languages, $\co\CC$ means the collection of all languages whose complements belong to $\CC$.

\subsection{Classical Finite Automata and Cut Point Formulation}\label{sec:classical-cutpoint}


We assume the reader's familiarity with \emph{2-way probabilistic finite automata} (or 2pfa's, in short) with real transition probabilities.
To make it easier to understand a direct connection to quantum finite automata, we formulate such 2pfa's as $(Q,\Sigma,\delta,q_0,Q_{acc},Q_{rej})$, by including a set $Q_{rej}$ of rejecting states, which was not present in  Rabin's original definition  in \cite{Rab63}. Formally, a 2pfa $M$  is a sextuple $(Q,\Sigma,\delta,q_0,Q_{acc},Q_{rej})$, in which a  tape head moves freely to the right, to the left, and stays still simply by applying a
transition function whose transition probabilities are drawn from $[0,1]$.  Moreover, $Q$ is a finite set of inner states, $q_0$ is the initial (inner) state, and $Q_{acc}$ is a set of accepting (inner) states. An input string $x=x_1x_2\cdots x_n$ of length $n$ is initially given onto an input tape, surrounded by two designated endmarkers $\cent$ (left endmarker) and $\dollar$ (right endmarker). Let $\check{\Sigma}=\Sigma\cup\{\cent,\dollar\}$. A tape head starts off at $\cent$ in the initial inner state $q_0$. For simplicity, all tape cells are indexed by integers from $0$ to $n+1$, where $\cent$ is located in cell $0$ and $\dollar$ is in cell $n+1$. The transition function $\delta: Q\times\check{\Sigma}\times Q\times D\to[0,1]$ with $D=\{0,\pm1\}$ naturally induces a \emph{transition matrix} acting on the vector space spanned by $\{(q,i)\mid q\in Q,i\in[0,|x|+1]_{\integer}\}$ and we demand that such a matrix should be
\emph{stochastic}.\footnote{A real square matrix is called \emph{stochastic} if every column of the matrix sums up to $1$. The use of ``columns'' instead of ``rows'' comes from the fact that we apply stochastic matrices from the right (not from the left), opposed to the initial formulation of  Rabin \cite{Rab63}, in accordance with the definition of quantum finite automata given in Section \ref{sec:quantum-finite-automata}.}
When all transition probabilities of $M$ are drawn from $K$,
we succinctly say that a 2pfa takes \emph{$K$-transition probabilities}. Implicitly, we always assume that $\{0,1/2,1\}\subseteq K$ so that $M$ can make any deterministic move and also flip any fair coin. The \emph{acceptance probability} (resp., \emph{rejection probability}) of $M$ on input $x$ is the sum of the probabilities that $M$ produces accepting (resp., rejecting) computation paths starting with the input $x$. Two notations $p_{M,acc}(x)$ and $p_{M,rej}(x)$ respectively denote the acceptance probability and the rejection probability of $M$ on $x$.
Occasionally, we write $\prob_{M}[M(x)=1]$ to express the probability of $M$ accepting $x$, and $\prob_{M}[M(x)=0]$ for the probability of $M$ rejecting $x$.

At this moment, it is important to discuss the acceptance criteria of 2pfa's. Since the work of Rabin \cite{Rab63}, the \emph{acceptance criteria} of a given probabilistic finite automaton are determined by a technical term of ``cut point,'' which is a threshold of its acceptance probabilities alone (neglecting rejection probabilities because non-accepting computation paths have been traditionally treated as ``rejected'').
Given a constant $\eta\in[0,1)$ and a language $L$,  a 2pfa $M$ is said to \emph{recognize $L$ with cut point} $\eta$ if (1) for any $x\in L$, $M$ accepts $x$ with probability more than $\eta$ (\ie $p_{M,acc}(x)>\eta$) and (2) for any $x\in\overline{A}$, $M$ accepts $x$ with probability at most $\eta$ (\ie $p_{M,acc}(x)\leq \eta$). It is known that we can set $\eta$ to be $1/2$ \cite{Tur68} by modifying the original 2pfa's properly. Similarly, we say that $M$ recognizes $L$ with \emph{isolated cut point} $\eta$ if there exists a constant $\varepsilon\in(0,1)$ with $0<\eta-\varepsilon\leq\eta+\varepsilon\leq1$ such that, for any $x\in L$, $p_{M,acc}(x)\geq \eta-\varepsilon$ and, for any $x\in\overline{A}$,  $p_{M,acc}(x)\leq \eta - \varepsilon$.

A language is \emph{$K$-stochastic} if it is recognized with an appropriate  cut point $\eta\in K\cap (0,1]$ by a certain \emph{1-way probabilistic finite automaton} (or a 1pfa) with $K$-transition probabilities, where a 1pfa\footnote{This machine is sometimes called a \emph{real-time} pfa.} always moves its tape head to the right until it scans $\dollar$ and halts.
When $K=\real$, we simply say that $L$ is \emph{stochastic}.  The notation $\mathrm{SL}_{K}$  refers to the family of all $K$-stochastic languages. Ka\c{n}eps \cite{Kan89} showed that the replacement of 1pfa's by 2pfa's does  not change the definition of $\mathrm{SL}_{\real}$. Moreover, the notation $\mathrm{SL}^{=}_{K}$ denotes the language family defined by the following criterion: there are a constant  (called an \emph{exact cut point}) $\eta\in K\cap (0,1]$ and a 1pfa $M$ with $K$-transition probabilities satisfying that, for all $x\in\Sigma^*$, $x \in L$ iff $p_{M,acc}(x)=\eta$.
When $K=\real$,  for example, it is possible to fix $\eta=1/2$. The complement family $\co\mathrm{SL}^{=}_{K}$ is sometimes denoted by $\mathrm{SL}^{\neq}_{K}$. It is not difficult to verify that $\co\mathrm{SL}_{\real}$ coincides with the family of all languages $L$ recognized by 1pfa's $M$ with ``non-strict cut points'' (which requires $p_{M,acc}(x)\geq \eta$ instead of $p_{M,acc}(x)>\eta$) for certain constants $\eta\in[0,1]$.
It is known in \cite{TYL10} that $\mathrm{SL}_{\rational}$ and $\mathrm{SL}^{=}_{\rational}$ are characterized in terms of one-tape linear-time Turing machines (namely, $\mathrm{1\mbox{-}PLIN}$ and $\mathrm{1\mbox{-}C_{=}LIN}$).
Despite our past efforts, we still do not know whether $\mathrm{SL}_{\real}$ is closed under complementation, whether $\mathrm{SL}^{=}_{\real}$ is included in $\mathrm{SL}_{\real}$, and whether $\mathrm{SL}^{=}_{\real}$ contains any non-recursive language (see, \eg \cite{Mac93} for references therein).

Regarding a 2pfa $M$, the \emph{expected running time} of $M$ on input $x$ is the average length of all computation paths produced during a computation of $M$ on $x$, provided that the probability of non-terminating computation paths is zero.
Opposed to this expected running time, we say that a 2pfa $M$ {\em runs in worst-case $t(n)$-time} if, on any input $x$, all computation paths (including both  accepting and rejecting paths) of $M$ must have length at most $t(|x|)$.

As a variant of 2pfa's, we define a \emph{$k$-head $2$-way probabilistic finite automaton} (or $k$head-2pfa, for brevity) by allowing a 2pfa to use $k$ tape heads that move separately along a single input tape \cite{Mac97}. The notation $\twoppfa_{K}(k\mbox{-}head)$ denotes the family of all languages recognized with cut points in $K\cap(0,1]$ by $k$head-2pfa's. In a similar way, $\mathrm{2C_{=}PFA}_{K}(k\mbox{-}head)$ is defined using ``exact cut points'' instead of the aforementioned ``cut points.''
We write $\mathrm{2PPFA}_K(k\mbox{-}head)[\polytime]$ (resp., $\mathrm{2C_{=}PFA}_K(k\mbox{-}head)[\polytime]$) for the class of all languages recognized with cut points (resp., exact cut points) in $K\cap(0,1]$ by $k$-head 2pfa's that run in worst-case polynomial time.

The notation $\#\twopfa_{K}$ expresses the collection of all \emph{stochastic functions}, which are of the form $p_{M,acc}$ for certain 2pfa's $M$ with $K$-transition probabilities (see \cite{Mac93} for the case of 1pfa's).
Similarly to $\mathrm{2PPFA}_K(k\mbox{-}head)[\polytime]$, we can expand $\#\twopfa_{K}$ to another function class $\#\mathrm{2PFA}_K(k\mbox{-}head)[\polytime]$. Moreover, let us recall a probabilistic complexity class $\mathrm{PL}$, which has been explained in Section \ref{sec:overview}.

In the deterministic case, we write $\reg$ for the family of all \emph{regular languages}, which are recognized by \emph{1-way deterministic finite automata} (or \emph{1dfa's}).  A \emph{2-way reversible finite automaton} (or 2rfa) is a 2-way deterministic finite automaton $(Q,\Sigma,\delta,q_0,Q_{acc},Q_{rej})$ whose transition function $\delta: Q\times\check{\Sigma}\to Q\times D$ satisfies the following \emph{reversibility property}: for any pair $p\in Q$ and $d\in D$, there exists a unique pair $(q,\sigma)\in Q\times \check{\Sigma}$ for which  $\delta(q,\sigma)=(p,d)$ holds. Let $\mathrm{2RFA}$ denote the family of all languages recognized by 2rfa's.

\subsection{Quantum Finite Automata and Bounded Error Formulation}\label{sec:quantum-finite-automata}

We briefly give the formal definition of \emph{2-way quantum finite automata}  (or 2qfa's, in short). Formally, a 2qfa $M$ is described as a sextuple $(Q,\Sigma,\delta,q_0,Q_{acc},Q_{rej})$, where $Q$ is a finite set of inner states with $Q_{acc}\cup Q_{rej}\subseteq Q$ and $Q_{acc}\cap Q_{rej}=\setempty$, $\Sigma$ is a finite alphabet, $q_0$ is the initial inner state, and $\delta$ is a transition function mapping from $Q\times\check{\Sigma}\times Q\times D$ to $\complex$, where $\check{\Sigma}$ and $D$ have been defined in the previous subsection.
The transition function $\delta$ describes a series of {\em transitions} and its values are called transition amplitudes (or \emph{amplitudes}). An expression $\delta(p,\sigma,q,d)=\gamma$  means that, assuming that the 2qfa $M$ is in inner state $p$ scanning a symbol $\sigma$, $M$ at the next step changes its inner state to $q$ and moves its tape head in direction $d$ with amplitude $\gamma$. The set $Q$ is partitioned into three sets: $Q_{acc}$, $Q_{rej}$, and $Q_{non}$. Inner states in $Q_{acc}$ (resp., in $Q_{rej}$) are called {\em accepting states} (resp., {\em rejecting states}). A {\em halting state} refers to an inner state in $Q_{acc}\cup Q_{rej}$. The rest of inner states, denoted by $Q_{non}$, consists of {\em non-halting states}. We say that $M$ has \emph{$K$-amplitudes} if all amplitudes of $M$ belong to set $K$ ($\subseteq\complex$).

Similarly to the case of 2pfa's, an input tape has two endmarkers $\cent$ and $\$$ and its tape cells are indexed by integers between $0$ and $n+1$ whenever a given input has length $n$.
For technical convenience, we additionally assume that the input tape is {\em circular} (as originally defined in \cite{KW97}).
A {\em (classical) configuration} is a description of a single moment (or a snapshot) of $M$'s computation, which is formally expressed as a pair of an inner state and a head position in $[0,n+1]_{\integer}$.
An application of $\delta$ can be viewed as an application of a linear operator over a configuration space. Given any input $x$ of length $n$, a \emph{configuration space} $\mathcal{CONF}_{n}$ is a Hilbert space spanned by $\{\qubit{q,\ell}\mid q\in Q,\ell\in[0,n+1]_{\integer}\}$.
From $\delta$ and input $x\in\Sigma^n$,
we define a \emph{time-evolution operator} $U_{\delta}^{(x)}$ as a
linear operator acting on the configuration space in the following way: for each $(p,i)\in Q\times[0,n+1]_{\integer}$, $U^{(x)}_{\delta}$ maps $\qubit{p,i}$ to  $\sum_{(q,d)\in Q\times\{0,\pm1\}}\delta(p,x_i,q,d)\qubit{q,i+d\,(\mathrm{mod}\,n+2)}$, where $x_0=\cent$, $x_{n+1}=\$$, and $x_i$ is the $i$th symbol of $x$ for each index  $i\in[1,n]_{\integer}$. Throughout this paper, we always assume $U_{\delta}^{(x)}$ to be \emph{unitary} for every string $x$.
Three projections $\Pi_{acc}$, $\Pi_{rej}$, and $\Pi_{non}$ are linear maps projecting onto the spaces $W_{acc} = span\{\qubit{q}\mid q\in Q_{acc}\}$, $W_{rej} = span\{\qubit{q}\mid q\in Q_{rej}\}$, and $W_{non} = span\{\qubit{q}\mid q\in Q_{non}\}$, respectively. A \emph{computation} of $M$ on input $x$ proceeds as follows. The 2qfa $M$ starts with its initial configuration $\qubit{\phi_0}=\qubit{q_0}\qubit{0}$ (where $0$ means that the tape head is scanning  $\cent$). At Step $i$, $M$ applies $U_{\delta}^{(x)}$ to $\qubit{\phi_{i-1}}$ and then applies $\Pi_{acc}\oplus \Pi_{rej}\oplus \Pi_{non}$. We say that $M$ \emph{accepts} (resp., \emph{rejects}) $x$ at Step $i$ with probability $p_{M,acc,i}(x)= \|\Pi_{acc}U_{\delta}^{(x)}\qubit{\phi_{i-1}}\|^2$ (resp., $p_{M,rej,i}(x)= \|\Pi_{rej}U_{\delta}^{(x)}\qubit{\phi_{i-1}}\|^2$). The $i$th quantum state $\qubit{\phi_i}$ is $\Pi_{non}U_{\delta}^{(x)}\qubit{\phi_{i-1}}$. The \emph{acceptance probability} $p_{M,acc}(x)$ of $M$ on $x$ is $\sum_{i=1}^{\infty}p_{M,acc,i}(x)$. The \emph{rejection probability} is defined similarly and is denoted by $p_{M,rej}(x)$.

With respect to acceptance criteria of 1fa's, we have customarily taken bounded-error and unbounded-error formulations. Let $\varepsilon$ be any constant in $[0,1/2)$ (called an \emph{error bound}) and let $L$ be any language  over alphabet $\Sigma$. We say that a 2qfa $M$ {\em recognizes $L$ with error probability at most $\varepsilon$} if (i) for every $x\in L$, $p_{M,acc}(x)\geq 1-\varepsilon$ and (ii) for every $x\in \overline{L}$ ($=\Sigma^*-L$), $p_{M,rej}(x)\geq 1-\varepsilon$.
When such an $\varepsilon$ exists, we customarily say that $M$ recognizes $L$ with \emph{bounded-error probability}.
We define the class $\twobqfa_{K}$ as the collection of all languages that can be recognized by bounded-error 2qfa's  with $K$-amplitudes.
It is important to note that these 2qfa's may not halt with certain probability up to $\varepsilon$.
Opposed to the bounded-error criterion, we say that $M$ recognizes $L$ with \emph{unbounded-error probability} if (i') for any $x\in L$, $M$ accepts $x$ with probability more than $1/2$ (i.e., $p_{M,acc}(x)>1/2$) and (ii') for any $x\notin L$, $M$ rejects $x$ with probability at least $1/2$ (i.e., $p_{M,rej}(x)\geq1/2$).
We then obtain the unbounded-error language family $\twopqfa_{K}$ as the collection of languages recognized by 2qfa's with unbounded-error probability.

Concerning halting computation, we say that a 2qfa {\em halts completely} if its halting probability equals $1$, whereas a 2qfa {\em halts absolutely} if all the computation  paths of the 2qfa eventually terminate in halting inner states. If a 2qfa halts absolutely, then it must halt completely, but the converse is not always true since a 2qfa that halts completely might possibly have a computation path that does not terminate.
When $M$ halts completely, the \emph{expected running time} of $M$ on $x$ is defined to be the average length of all computation paths.

To place various restrictions, specified as $\pair{restrictions}$, on 2qfa's, we generally use a conventional notation of the form $\twobqfa_{K}(restrictions)$. For example, two restrictions $\pair{\comphalt}$ and $\pair{\abshalt}$ respectively indicate that a 2qfa halts completely and absolutely.
Another restriction $\pair{\lintime}$ means that a 2qfa runs in expected liner time. More generally, $\twobqfa_{K}(t(n)\mbox{-}time)$ is defined by $K$-amplitude 2qfa's which run in expected time at most $t(n)$ (that is, the average running time of $M$ on each input of length $n$ is bounded from above by $t(n)$).

We shall discuss four more language families. The error-free language family $\twoeqfa_{K}$ is obtained from $\twobqfa_{K}$ by setting $\varepsilon=0$ (\ie either $p_{M,acc}(x)=1$ or $p_{M,rej}(x)=1$ for all $x\in\Sigma^*$).
The one-sided error language family $\tworqfa_{K}$ requires the existence of an error-bound $\varepsilon\in[0,1/2)$ such that $p_{M,acc}(x)\geq 1-\varepsilon$ for all $x\in L$ and $p_{M,rej}(x)=1$ for all $x\in\overline{L}$.

In contrast, the equality language family $\twocequalqfa_{K}$ is composed of languages $L$ recognized by $K$-amplitude 2qfa's $M$ with nonnegative exact cut points; namely, there exists a constant $\eta \in K\cap(0,1]$ such that, for every $x$,  $x\in L$ iff  $p_{M,acc}(x)= \eta$.
All languages $L$ recognized by $K$-amplitude 2qfa's $M$ with zero cut point  forms the nondeterministic language family $\twonqfa_{K}$, i.e., for every $x$, $x\in L$ iff $p_{M,acc}(x)>0$.
From those definitions of language families follow a series of natural properties. See also \cite[Lemma 4.9]{Yam03} for comparison.

\begin{lemma}\label{basic-inclusion}
Let $K$ be any nonempty subset of $\complex$ with $\{0,1/2,1\}\subseteq K$.
\begin{enumerate}\vs{-1}
  \setlength{\topsep}{-2mm}%
  \setlength{\itemsep}{1mm}%
  \setlength{\parskip}{0cm}%

\item $\twoeqfa_{K} \subseteq \tworqfa_{K} \subseteq \twobqfa_{K} \subseteq \twopqfa_{K}$.

\item $\tworqfa_{K}\subseteq \twonqfa_{K}$.

\item $\twoeqfa_{K}=\co\twoeqfa_{K}$ and $\twobqfa_{K}=\co\twobqfa_{K}$.

\item $\twoeqfa_{K} = \twoeqfa_{K}(\comphalt) \subseteq \twocequalqfa_{K}(\comphalt) \cap \co\twocequalqfa_{K}(\comphalt)$.

\item $\tworqfa_{K}\cup\co\tworqfa_{K}\subseteq \twobqfa_{K}$ if $\rational\cap[0,1]\subseteq K$.

\item $\twonqfa_{K}(\comphalt) \subseteq \co\twocequalqfa_{K}(\comphalt)$.
\end{enumerate}
\end{lemma}

\begin{proof}
(1)--(3) Trivial from the definitions of those classes. In particular, (3) follows easily by exchanging $Q_{acc}$ and $Q_{rej}$ in the definition of qfa's.

(4) The first equality is obvious from the requirement for the error-free property of $\twoeqfa_{K}$.
Given a language $L\in\twoeqfa_{K}$, take a completely-halting 2qfa $M$ recognizing $L$ with zero error and $K$-amplitudes. Let us define another 2qfa $N$ that starts simulating $M$ on input $x$. Whenever $M$ halts with acceptance, we wish to make $N$ enter both accepting and rejecting states with equal probability $1/2$. However, since we need to restrict $N$'s amplitudes within $\{0,1/2,1\}$, we employ the following simple trick. The machine $N$ prepares fresh $4$ inner states, say, $\{q'_1,q'_2,q'_3,q'_4\}$ and enters each of those inner states with equal amplitude $1/2$. We associate the first two inner states with accepting states and the last two inner states with rejecting
states.
In contrast, when $M$ halts with rejection, $N$ simply rejects $x$ with probability $1$.

When $x\in L$, since $p_{M,rej}(x)=0$, we obtain $p_{N,acc}(x)= p_{N,rej}(x) = 1/2$; on the contrary, when $x\notin L$, $p_{N,rej}(x)=1$ and $p_{M,acc}(x)=0$. From these relations, we conclude that $L$ belongs to  $\twocequalqfa$. Since $\twoeqfa_{K}$ is closed under complementation by (2),
the inclusion $\co\twoeqfa_{K}\subseteq \twocequalqfa_{K}$ also follows.

(5) First, we shall show that every language $L$ in $\tworqfa_{K}$ is also a member of $\twobqfa_{K}$. Take a one-sided-error 2qfa $M = (Q,\Sigma,\delta,Q_{acc},Q_{rej})$ that recognizes $L$ with error bound $\varepsilon\in[0,1/2]$. If $\varepsilon<1/2$, then $L$ is in $\twobqfa_{K}$. Next, let us consider the remaining case of $\varepsilon=1/2$.

Let us define a new 2qfa $N$ as follows. Choose a real number $\alpha$ for which $0< \alpha < 1/2$ and $\{\sqrt{\alpha},\sqrt{1-\alpha}\}\subseteq \rational$, and define $\varepsilon' = \frac{1-\alpha}{2}$. Clearly, $\alpha\leq \varepsilon' <1/2$ holds. Given an input $x$, $N$ starts with the initial configuration $\qubit{q_0}\qubit{0}$. On scanning $\cent$, $N$ transforms $\qubit{q_0}\qubit{0}$ into $\sqrt{\alpha}\qubit{q'_{rej}}\qubit{0} + \sqrt{1-\alpha}{\;} U_{\delta}^{(x)}\qubit{q_0}\qubit{0}$, where $q'_{rej}$ is a  fresh rejecting state. The first term is traced out immediately by a measurement. In contrast, the second term evolves as $N$ applies $U_{\delta}^{(x)}$.
When $x\in L$, since $p_{M,acc}(x)=1$, the acceptance probability $p_{N,acc}(x)$ of $N$ on the input $x$ satisfies that $p_{N,acc}(x) = (1-\alpha) p_{M,acc}(x) = 1-\alpha\geq 1-\varepsilon'$.
On the contrary, when $x\notin L$, $p_{M,rej}(x) = \frac{1}{2}$ implies that  the rejection probability $p_{N,rej}(x)$ is $\alpha + (1-\alpha) p_{M,rej}(x)$, which is at least $1-\frac{1}{2}(1-\alpha) \geq 1-\varepsilon'$. We then conclude that $L$ belongs to $\twobqfa_{K}$ because $K$ contains $\rational\cap[0,1]$.

Since $\twobqfa_{K}$ is closed under complementation by (2), it follows from the first containment that $\co\tworqfa_{K}\subseteq \twobqfa_{K}$. As a result, we obtain $\tworqfa_{K}\cup \co\tworqfa_{K}\subseteq \twopqfa_{K}$.

(6) Let $L$ be any language in $\twonqfa_{K}(\comphalt)$ and let $M$ denote a completely-halting 2qfa recognizing $L$ with cut point $0$.
A new 2qfa $N$ is constructed from $M$ to behave as follows. On input $x$, $N$ simulates $M$ on $x$ and, when $M$ enters its rejecting state, $N$ instead enters two accepting states and two rejecting states with equal amplitudes $1/2$, as in (4). It then follows that $p_{N,acc}(x)$ equals $p_{M,acc}(x) + \frac{1}{2} p_{M,rej}(x)$, which turns out to be $\frac{1}{2}(1+p_{M,acc}(x))$ since $p_{M,acc}(x)+p_{M,rej}(x)=1$. Therefore, the membership $x\in L$ implies $p_{N,acc}(x)\neq 1/2$ because of $p_{M,acc}(x)>0$. When $x\notin L$, on the contrary, $p_{M,acc}(x)=0$ leads to $p_{N,acc}(x)=1/2$. Therefore, $L$ belongs to $\co\twocequalqfa_{K}(\comphalt)$.
\end{proof}

Notice that, when a 2qfa $M$ completely halts, the bounded-error criterion of $M$ coincides with the isolated cut point criterion, because  $p_{M,acc}(x)+p_{M,rej}(x)=1$ holds for all $x\in\Sigma^*$. Therefore, it is possible to define $\twopqfa(\comphalt)$ and $\twocequalqfa(\comphalt)$ in a slightly different way.

\begin{lemma}\label{classical-quantum}
For any language $L$ over alphabet $\Sigma$, $L$ is in $\twopqfa_{K}(\comphalt)$ (resp., $\twocequalqfa_{K}(\comphalt)$) iff there exist a $K$-amplitude 2qfa $M$
that completely halts and satisfies that, for all $x\in\Sigma^*$,  $x\in L$ iff  $p_{M,acc}(x)> p_{M,rej}(x)$ (resp., $p_{M,acc}(x)= p_{M,rej}(x)$). Moreover, the same is true for worst-case linear-time 2qfa's.
\end{lemma}

\begin{proof}
Consider a completely-halting 2qfa $M$. Since $M$ completely halts, it follows that $p_{M,acc}(x)+p_{M,rej}(x)=1$ for all $x$. Thus, we conclude that $p_{M,acc}(x)>p_{M,rej}(x)$ (resp., $p_{M,acc}(x)=p_{M,rej}(x)$) iff $p_{M,acc}(x)>1/2$ (resp., $p_{M,acc}(x)=1/2$). Hence, the lemma for $\twopqfa_{K}(\comphalt)$ and $\twocequalqfa_{K}(\comphalt)$ follows instantly.
\end{proof}

The following folklore lemma helps us concentrate on real amplitudes when we discuss unrestricted-amplitude 2qfa's.

\begin{lemma}\label{folke-complex}{\rm (folklore)}
Every $\complex$-amplitude 2qfa can be simulated by a certain $\real$-amplitude 2qfa with the same acceptance/rejection/non-halting probabilities using only twice the number of original inner states such that its tape head moves are exactly the same as the original 2qfa's.
\end{lemma}

\begin{proofsketch}
Let $M=(Q,\Sigma,\delta,q_0,Q_{acc},Q_{rej})$ be any $\complex$-amplitude 2qfa. We shall define another 2qfa $N=(Q',\Sigma,\delta',q'_0,Q'_{acc},Q'_{rej})$ with the desired property. Let $Q' = Q\times \{I,R\}$, $Q'_{acc}= Q_{acc}\times\{I,R\}$, and $Q'_{rej} = Q_{rej}\times\{I,R\}$. Moreover, let $q'_0 = (q_0,R)$ and let
\[
\delta'((q,b),\sigma,(p,c),d) = \left\{ \begin{array}{ll} \mathrm{Re}(\delta(q,\sigma,p,d)) & \text{if $b=c\in\{I,R\}$,} \\
\mathrm{Im}(\delta(q,\sigma,p,d)) & \text{if $b=R$ and $c=I$,} \\
-\mathrm{Im}(\delta(q,\sigma,p,d)) & \text{if $b=I$ and $c=R$.}
\end{array} \right.
\]
It follows that $\delta(q,\sigma,p,d)$ equals $\delta'((q,R),\sigma,(p,R),d)
 + \imath \cdot \delta'((q,R),\sigma,(p,I),d)$, which further equals $\delta'((q,I),\sigma,(p,I),d)
 - \imath \cdot \delta'((q,I),\sigma,(p,R),d)$.
From those equalities, we can derive the desired conclusion.
\end{proofsketch}

As an immediate consequence of Lemma \ref{folke-complex}, we obtain, for example, $\twobqfa_{\complex} = \twobqfa_{\real}$ and $\twobqfa_{\complex}(t(n)\mbox{-}time) = \twobqfa_{\real}(t(n)\mbox{-}time)$ for any time bound $t(n)$.

In the case of one-way model, each qfa always moves its tape head to the right without stopping it and, after scanning $\dollar$, the qfa must ``halt.'' We call such a machine a \emph{one-way quantum finite automaton} (or \emph{1qfa}). Obviously, 1qfa's halt absolutely. Regarding language families, we define $\onepqfa_{K}$ and $\onecequalqfa_{K}$ by replacing underlying 2qfa's in the definition of $\twopqfa_{K}$ and $\twocequalqfa_{K}$ with 1qfa's, respectively.
At this point, we remark that 1pfa's with positive cut points also satisfy the unbounded-error criterion. To see this fact, let $M = (Q,\Sigma,\delta,q_0,Q_{acc},Q_{rej})$ denote a 1pfa recognizing language $L$ with cut point $\eta\in(0,1]$. As noted in Section \ref{sec:classical-cutpoint}, it is possible to set $\eta=1/2$ (thus,  $p_{M,acc}(x)>1/2$ for any $x\in L$, and  $p_{M,acc}(x)\leq 1/2$ for any $x\notin L$). By our convention of one-way head moves of 1pfa's,  $M$ must halt by the time when the right endmarker $\dollar$ is read. We transform  $M$ by, after reading $\dollar$, redirecting all inner states in $Q-Q_{acc}$ to a new unique rejecting state. In the end, it holds that, for all $x\in L$, $p_{M,acc}(x)>1/2$ and, for all $x\not\in L$, $p_{M,rej}(x)\geq 1/2$.

Lately, Yakary{\i}lmaz and Say \cite{YS10,YS11} discovered that  $\onepqfa_{\complex}$ and $\onenqfa_{\complex}$ precisely characterize  $\mathrm{SL}_{\real}$ and $\mathrm{SL}^{\neq}_{\real}$, respectively. In a similar way,   $\onecequalqfa_{\complex}$ can be shown to coincide with $\mathrm{SL}^{=}_{\real}$.

\begin{lemma}\label{PQFA-equal-SL}
$\mathrm{SL}_{\real} = \onepqfa_{\complex}$, $\mathrm{SL}^{\neq}_{\real} = \onenqfa_{\complex}$, and $\mathrm{SL}^{=}_{\real} =  \onecequalqfa_{\complex}$.
\end{lemma}

\begin{proof}
For the first statement, let us take a 1pfa with a positive cut point. As noted above, we can transform it to satisfy the bounded-error criterion. Let $M$ be the resulted 1pfa.
We then turn this 1pfa $M$ into an ``equivalent'' 1qfa, say, $N$ as
in \cite[Lemma 5.1]{YS11} by embedding each stochastic matrix induced by $M$'s transition function $\delta$ into a larger-dimensional unitary transition matrix in such a way that, if  a quantum state that is produced by this matrix does not correctly represent the outcome of $\delta$, it is  mapped into both accepting and rejecting states with equal probability.
As a consequence, it holds that
$p_{M,acc}(x) = \eta(|x|) p_{N,acc}(x)+\frac{1-\eta(|x|)}{2}$ and  $p_{M,rej}(x) = \eta(|x|) p_{N,rej}(x)+\frac{1-\eta(|x|)}{2}$ for an appropriately chosen  positive function $\eta$ (i.e., $\eta(n)>0$ for all $n\in\nat^{+}$). These  equations lead to the conclusion that $p_{M,acc}(x)>1/2$ (resp., $p_{M,rej}(x)\geq 1/2$) iff $p_{N,acc}(x)>1/2$ (resp., $p_{N,rej}(x)\geq 1/2$). Thus, $\mathrm{SL}_{\real}\subseteq \onepqfa_{\complex}$ follows.

By essentially the same idea as above, it was shown in \cite[Lemma 2]{YS10} that, for any 1pfa $M$, there are a positive function $\eta$ and a 1qfa $N$ such that $p_{N,acc}(x)=\eta(|x|)(\frac{2p_{M,acc}(x)-1}{4})^2$ and $p_{N,rej}(x) = \eta(|x|) (\frac{3-p_{M,acc}(x)}{4})^2$ for all $x$. From these equalities, we conclude that $p_{M,acc}(x)\neq 1/2$ iff $p_{N,acc}(x)>0$. This instantly yields $\mathrm{SL}^{\neq}_{\real}\subseteq \onenqfa_{\complex}$.
For a similar reason as in the proof of Lemma \ref{basic-inclusion}(6), we obtain $\onenqfa_{\complex}\subseteq \co\onecequalqfa_{\complex}$; hence, $\mathrm{SL}^{=}_{\real}\subseteq \onecequalqfa_{\complex}$
immediately follows.

Conversely, consider any 1qfa $N$ that recognizes $L$ with unbounded-error probability. We can assume, by Lemma \ref{folke-complex}, that $N$ uses only $\real$-amplitudes. As shown in \cite[Lemma 3.1]{YS11}, there exists a 1pfa $M$ satisfying that $p_{N,acc}(x) = p_{M,acc}(x)$ and $p_{N,rej}(x) = p_{M,rej}(x)$ for all $x$. From this fact, we can derive the following three inclusions:   $\onepqfa_{\complex}\subseteq \mathrm{SL}_{\real}$, $\onenqfa_{\complex}\subseteq \mathrm{SL}^{\neq}_{\real}$, and $\onecequalqfa_{\complex}\subseteq \mathrm{SL}^{=}_{\real}$.
\end{proof}

From Lemma \ref{PQFA-equal-SL} follow two natural separations between the 1-way model and the 2-way model of quantum finite automata.

\begin{corollary}\label{onePQFA-included-2PQFA}
$\onepqfa_{\complex} \subsetneqq \twopqfa_{\complex}[\lintime]$ and
$\mathrm{1C_{=}QFA}_{\complex}\subsetneqq \twocequalqfa_{\complex}[\lintime]$.
\end{corollary}

\begin{proof}
Since 1qfa's outcomes are determined by the time their tape heads scan $\dollar$, it is possible to modify the 1qfa's by adding extra 2-way transitions that, after scanning $\dollar$, map each non-halting inner state to both a fresh accepting state and a fresh rejecting state with equal probability.
It thus follows that $\onepqfa_{\complex}\subseteq\twopqfa_{\complex}[\lintime]$. As for $\onecequalqfa_{\complex}$,  after scanning $\dollar$, it suffices to map all non-halting states of 1qfa's to fresh rejecting states. This mapping derives the inclusion   $\onecequalqfa_{\complex}\subseteq \twocequalqfa_{\complex}[\lintime]$.

Next, we denote by $L_{NH}$ the special language $\{a^{m}ba^{k_1}b\cdots ba^{k_d}b\mid m,k_1,\ldots,k_d\in\nat^{+},\exists i\in[1,d]_{\integer}[m=k_1+\cdots +k_i]\}$ over alphabet $\{a,b\}$ \cite{NH71}. From a result of Freivalds and Karpinski \cite{FK94} follows a non-membership $L_{NH}\notin \mathrm{SL}_{\real}$ \cite{FK94}. This also yields $L_{NH}\notin \mathrm{SL}^{\neq}_{\real}$ because $\mathrm{SL}^{\neq}_{\real} \subseteq \mathrm{SL}_{\real}$.
In contrast, the proof of \cite[Theorem 4.1]{YS11} actually shows that $\overline{L_{NH}}$ can be recognized by a certain qfa, say, $N$
whose tape head either moves to the right or stays still (such a qfa is known as a \emph{1.5-way qfa}) with exact cut point $1/2$. Hence, $L_{NH}$ belongs to $\co\twocequalqfa_{\algebraic}[\lintime]$, which is a subclass of $\twopqfa_{\algebraic}[\lintime]$. As a consequence, we obtain $\mathrm{SL}_{\real} \neq \twopqfa_{\complex}(\abshalt)$ and $\mathrm{SL}_{\real}^{\neq} \neq \co\twocequalqfa_{\complex}(\abshalt)$; thus, Lemma \ref{PQFA-equal-SL} implies that $\onepqfa_{\complex}\neq\twopqfa_{\complex}[\lintime]$ and $\onecequalqfa_{\complex}\neq \twocequalqfa_{\complex}[\lintime]$.
\end{proof}

The notion of \emph{quantum functions} generated by quantum Turing machines, given in \cite{Yam03}, is quite useful in describing various language families.  Similarly to a quantum-function class $\#\mathrm{QP}_{K}$ in \cite{Yam03}, we define $\#\twoqfa_{K}$ to be the set of (quantum) functions $p_{M,acc}:\Sigma^*\to[0,1]$ for all $K$-amplitude 2qfa's $M$. Such functions may be seen as an extension of \emph{stochastic functions} of Macarie \cite{Mac93}.
Note that, by exchanging $Q_{acc}$ and $Q_{rej}$ of $M$, $p_{M,rej}$ also belongs to $\#\twoqfa_{K}$.

As a natural analogue of multi-head 2pfa's, we introduce a two-head model of quantum finite automata, first introduced in \cite{ABF+99} as ``multi-tape'' quantum finite automata. This machine model is defined by a transition function of the form $\delta:Q\times\check{\Sigma} \times\check{\Sigma} \times D\times D\to\complex$.  Let $U_{\delta}^{(x)}$ be a time-evolution matrix acting on the configuration space
$span\{\qubit{q,h_1,h_2}\mid q\in Q, h_1,h_2\in [0,n+1]_{\integer}\}$ defined as
\[
U_{\delta}^{(x)}\qubit{q,h_1,h_2} = \sum_{(q',d_1,d_2)\in Q\times D^2} \delta(q,x_{h_1},x_{h_2},q',d_1,d_2) \qubit{q',h_1+d_1\;(\mathrm{mod}\;n+2),h_2+d_2\;(\mathrm{mod}\;n+2)},
\]
where $\cent x\dollar = x_0x_1\cdots x_nx_{n+1}$.
To make $M$ \emph{well-formed}, we need to demand that $U_{\delta}^{(x)}$ should be unitary. We can further generalize this 2-head model to a $k$-head model for any index $k\geq2$.

Briefly, we shall discuss simple properties of several functional operations among quantum functions taken from $\#\twoqfa_{K}$. As is shown in the next lemma, unlike $\#\mathrm{QP}_{K}$, the function class $\#\twoqfa_{K}$ does not seem to enjoy various closure properties, which $\#\mathrm{QP}_{K}$ naturally enjoys (see \cite{Yam03}).

\begin{lemma}\label{quantum-functions}
Let $f,g\in\#\twoqfa_{K}$ and let $\alpha,\beta\in\real\cap[0,1]$.
\begin{enumerate}\vs{-1}
  \setlength{\topsep}{-2mm}%
  \setlength{\itemsep}{1mm}%
  \setlength{\parskip}{0cm}%

\item If $f\in \#\twoqfa_{K}(\comphalt)$, then $1-f\in\#\twoqfa_{K}(\comphalt)$.

\item If $\alpha+\beta\leq1$ and $\sqrt{\alpha},\sqrt{\beta},\sqrt{1-\alpha-\beta}\in K$, then $\alpha f+\beta g\in \#\twoqfa_{K}(2\mbox{-}head)$.

\item $f\cdot g\in \#\twoqfa_{K}(2\mbox{-}head)$.
\end{enumerate}
\end{lemma}

\begin{proof}
Let $f,g\in\#\twoqfa_{K}$ and assume that 2qfa's $M_f$ and $M_g$ witness $f$ and $g$, respectively; namely, $f(x) = p_{M_f,acc}(x)$ and $g(x) = p_{M_g,acc}(x)$ for all strings $x$. Let $M_f=(Q_f,\Sigma,\delta_f,r_f,Q_{f,acc},Q_{f,rej})$ and  $M_g=(Q_g,\Sigma,\delta_g,r_g,Q_{g,acc},Q_{g,rej})$ satisfying $Q_{f}\cap Q_{g}=\setempty$, for our convenience.

(1) This is immediate from the fact that $p_{M,acc}(x)+p_{M,rej}(x)=1$ for all $x$ if $M$ halts completely.

(2) The desired $2$-head machine $N$ for $\alpha f+ \beta g$ takes input $x$, starts with the initial configuration $\qubit{q_0}\qubit{0}\qubit{0}$, and transforms it to $(\sqrt{\alpha}\qubit{r_f}+ \sqrt{\beta}\qubit{r_g}+ \sqrt{1-\alpha-\beta}\qubit{r_{rej}}) \qubit{0}\qubit{0}$ by moving the second tape head leftward and then rightward. More precisely, we set $\delta(q_0,\cent) = \qubit{q_1,0,-1}$ and $\delta(q_1,\dollar) = \sqrt{\alpha}\qubit{r_f,0,+1} + \sqrt{\beta}\qubit{r_g,0,+1} + \sqrt{1-\alpha-\beta}\qubit{r_{rej},0,+1}$.

If $r_f$ is observed, then we run $M_f$ starting with $r_f$ as its initial state using only the first tape head.
Similarly, when $r_g$ is observed, we run $M_g$ instead.
Let $Q_{acc} = Q_{f,acc}\times\{r_f\}$ and $Q_{rej} = Q_{f,rej}\cup Q_{g,rej}$.
The above procedure defines a well-formed 2head-2qfa. It is not difficult to show that $N$'s output $p_{N,acc}(x)$ equals $\alpha p_{M_f,acc}(x) + \beta p_{M_g,acc}(x)$ for any input $x$.

(3) Starting with $\qubit{r_f}\qubit{r_g}$ on input $x$, the desired machine $N$ simulates $M_f$ and $M_g$ in a tensor product form; that is, $M_f$ uses the first register starting in inner state $r_f$ and $M_g$ uses the second register in its initial state $r_g$. We then obtain a tensor product of two quantum states $\sum_{q\in Q_{f,acc}}\gamma_q\qubit{q}+\sum_{q\notin Q_{f,acc}}\eta_q\qubit{q}$ and $\sum_{q\in Q_{g,acc}}\gamma'_q\qubit{q}+\sum_{q\notin Q_{g,acc}}\eta'_q\qubit{q}$ for appropriate amplitudes $\gamma_q$, $\gamma'_q$, $\eta_q$, and $\eta'_q$. The machine $N$ uses $Q_{acc} = Q_{f,acc}\times Q_{g,acc}$ and $Q_{rej} = (Q_{f,acc}\times Q_{g,rej})\cup (Q_{f,rej}\times Q_{g,acc})$. The overall acceptance probability of $N$ on $x$ is $\| \sum_{q\in Q_{f,acc}}\gamma_q\qubit{q}\otimes \sum_{q\in Q_{g,acc}}\gamma'_q\qubit{q} \|^2$, which clearly equals $f(x)g(x)$.
\end{proof}

\section{Termination Criteria of Quantum Finite Automata}\label{sec:runtime-2qfaing}

In a stark contrast with 1qfa's, 2qfa's are, in general, not always guaranteed to halt (in finite steps); even bounded-error 2qfa's may produce computation paths that do not terminate. What will happen if we place different termination conditions on all computation paths of 2qfa's?
In this section, we wish to discuss such an issue regarding the effect of various termination criteria of 2qfa's. In particular, we shall investigate two specific cases of 2qfa's: absolutely-halting 2qfa's and completely-halting 2qfa's.

\subsection{Behaviors of Absolutely Halting QFAs}\label{sec:absolute-halt}

We begin with 2qfa's that \emph{halt absolutely} (that is,  all non-zero amplitude computation paths of a given qfa halt on all inputs within a finite number of steps). Since those 2qfa's are relatively easy to handle, it is possible for us to  obtain certain intriguing properties of them. Through this section, we shall describe those properties in details.
In what follows, we write $\am(2pfa,\polytime)$ for the family of all languages recognized by Dwork-Stockmeyer \emph{interactive proof systems} using 2pfa verifiers with $\real$-transition probabilities running in \emph{expected   polynomial time} \cite{DS92a}. For the formal definition and basic properties of this particular language family, refer to \cite{DS92a,NY09}.

\begin{proposition}\label{reg-qfa-am}
$\reg\subseteq \twoeqfa_{\rational}(\abshalt)\subseteq \co\tworqfa_{\rational}(\abshalt)\nsubseteq \am(2pfa,\polytime)$.
\end{proposition}

\begin{proof}
Since all 1dfa's can be simulated by certain 2rfa's that halt absolutely \cite[Corollary 5]{KW97}, we immediately conclude that $\reg\subseteq \twoeqfa_{\rational}(\abshalt)$. Similarly to Lemma \ref{basic-inclusion}(1), it follows  that $\twoeqfa_{\real}(\abshalt)\subseteq \tworqfa_{\rational}(\abshalt)\cap \co\tworqfa_{\rational}(\abshalt)$.
As for the last separation of the proposition,
Dwork and Stockmeyer \cite{DS92a} earlier showed that the language $UPal=\{0^n1^n\mid n\geq1\}$ over the binary alphabet $\{0,1\}$ does not belong to $\am(2pfa,\polytime)$.  However, since $UPal$ is in $\co\tworqfa_{\rational}(\abshalt)$ \cite[Proposition 2]{KW97}, we obtain the desired separation.
\end{proof}

Next, we shall give a precise bound of the running time of the 2qfa's when they halt absolutely. For convenience, we say that a 2qfa  halts in \emph{worst-case linear time} if every computation path of the 2qfa terminates within time linear in input size. In this case, we use another notation $\twobqfa_{K}[\lintime]$  for the family of languages witnessed by such 2qfa's using only $K$-amplitudes  to differentiate from the case of \emph{expected liner-time} computation, in which some computation paths may not even terminate. Similar bracketed notations can be introduced for $\twoeqfa$,  $\tworqfa$, $\twocequalqfa$, and $\twopqfa$. In the next theorem, we prove that every absolutely-halting 2qfa actually terminates in worst-case linear time.

\begin{theorem}\label{abs-halt-vs-lin-time}
For any nonempty set $K\subseteq\complex$, $\twobqfa_{K}(\abshalt) = \twobqfa_{K}[\lintime]$. The same is true for  $\twoeqfa$, $\tworqfa$,  $\twocequalqfa$, and $\twopqfa$.
\end{theorem}

By the definition of ``worst-case linear time,'' it naturally follows that $\twopqfa_{K}[\lintime]\subseteq \twobqfa_{K}(\abshalt)$. Hence, it suffices to focus on the proof of the other direction.
For this proof, we need to examine the behaviors of 2qfa's that halt absolutely.
Back in 1998, Yao \cite{Yao98} made the following useful observation regarding the length of their computation paths.

\begin{lemma}\label{claim-Yao98}\hs{1}
Any $\complex$-amplitude 2qfa with a set $Q$ of inner states should halt within worst-case $|Q|(n+2)+1$ steps if all (non-zero amplitude) computation paths of the 2qfa eventually terminate, where $n$ refers to input length.
\end{lemma}

Let us postpone the proof of Lemma \ref{claim-Yao98} for a while and we first demonstrate how to prove Theorem \ref{abs-halt-vs-lin-time} using this lemma.

\vs{-2}
\begin{proofof}{Theorem \ref{abs-halt-vs-lin-time}}
Let $L$ be any language in $\twobqfa_{K}(\abshalt)$, which is recognized by a certain $K$-amplitude  2qfa, say, $M$ with bounded-error probability. Assume that $M$ halts absolutely; that is, all (non-zero amplitude) computation paths of $M$ on every input $x$ eventually terminate in finitely many steps. By Lemma \ref{claim-Yao98}, we conclude that $M$ halts within worst-case $O(n)$ steps. This conclusion indicates that $L$ belongs to $\twobqfa_{K}[\lintime]$.
Thus, we immediately obtain $\twobqfa_{K}(\abshalt)\subseteq \twobqfa_{K}[\lintime]$. Since the converse containment is trivial, it follows that $\twobqfa_{K}(\abshalt) = \twobqfa_{K}[\lintime]$, as requested.
\end{proofof}

Lemma \ref{claim-Yao98} is so useful that it leads to not only Theorem \ref{abs-halt-vs-lin-time} but also various other consequences.

\begin{corollary}\label{amplitude-reduction}
\begin{enumerate}\vs{-1}
  \setlength{\topsep}{-2mm}%
  \setlength{\itemsep}{1mm}%
  \setlength{\parskip}{0cm}%

\item $\twopqfa_{\algebraic}(\abshalt) = \co\twopqfa_{\algebraic}(\abshalt)$.

\item $\twocequalqfa_{\algebraic}(\abshalt)\cup \co\twocequalqfa_{\algebraic}(\abshalt) \subseteq \twopqfa_{\algebraic}(\abshalt)$.

\item $\twoeqfa_{\complex}(\abshalt) = \twoeqfa_{\algebraic\cap\real}(\abshalt)$.
\end{enumerate}
\end{corollary}

In Corollary \ref{amplitude-reduction}(1--2),  we do not know at present whether the amplitude set $\algebraic$ can be replaced by $\complex$ because  the proof given below heavily relies on the property of algebraic numbers.

\begin{proofof}{Corollary \ref{amplitude-reduction}}
(1)
Let $L$ be any language in $\twopqfa_{\algebraic}(\abshalt)$ recognized with bounded-error probability by a certain $\algebraic$-amplitude 2qfa $M$ of the form $(Q,\Sigma,\delta,q_0,Q_{acc},Q_{rej})$. Theorem \ref{abs-halt-vs-lin-time} ensures that all computation paths of $M$ on any input of length $n$ terminate within $|Q|(n+2)+1$ steps. We write $F_{M}$ for a set of all transition amplitudes used in $M$. Since $F_{M}\subseteq \algebraic$ holds be the choice of $M$, we can choose numbers $\alpha_1,\alpha_2,\ldots,\alpha_e\in \algebraic$ so  that $F_{M}\subseteq \rational(\alpha_1,\ldots,\alpha_e)/\rational$, where $e=|Q||\Sigma||D|$.

Given an arbitrary input $x\in\Sigma^*$, we set $\alpha_{x} = p_{M,acc}(x)-1/2$ if $x\in L$, and  $\alpha_x= p_{M,rej}(x)-1/2$ otherwise. Let us consider $U^{(x)}_{\delta}$ and $\Pi_{non}$ associated with $M$.
Here, we claim that $p_{M,acc}(x)$ has the form $\sum_{k}a_{k}\left( \prod_{i=1}^e\alpha_i^{k_i} \right)$, where $k =(k_1,\ldots,k_e)$ ranges over $\integer_{[N_1]}\times\cdots\times\integer_{[N_e]}$,   $(N_1,\ldots,N_e)\in\nat^{e}$ with $N_i=2|Q|(n+2)+2$, and $a_{k}\in\integer$.
To verify this claim, we first note that $p_{M,acc}(x)$ is calculated as   $\sum_{t=0}^{|Q|(n+2)+1} \sum_{q\in Q_{acc}} \sum_{\ell\in[0,n+1]_{\integer}} | \bra{q,\ell} (U_{\delta}^{(x)}\Pi_{non})^t \ket{q_0,0}|^2$ by Lemma \ref{claim-Yao98}. Since each value $| \bra{q,\ell} (U_{\delta}^{(x)}\Pi_{non})^t \ket{q_0,0}|^2$ is written in the form  $\sum_{k}a'_{k}\left( \prod_{i=1}^e\alpha_i^{k_i} \right)$, which is a polynomial in $(\alpha_1,\ldots,\alpha_e)$, it is also possible to express $p_{M,acc}(x)$ as a polynomial in $(\alpha_1,\ldots,\alpha_e)$.
Hence, when $x\in L$, since $\alpha_x = p_{M,acc}(x)-1/2$, $\alpha_x$ can be expressed in a similar polynomial form. The case of $p_{M,rej}(x)$ is similar.

Next, we use the following known result taken from Stolarsky's textbook \cite{Sto74}. For other applications of this result, see \cite{Yam03,Yam11} for example.

\begin{lemma}\label{lemma-Sto74}\hs{2}
Let $\alpha_1,\ldots,\alpha_e\in\algebraic$. Let $h$ be the degree  of  $\rational(\alpha_1,\ldots,\alpha_e)/\rational$. There exists a constant $c>0$ that satisfies the following for any complex number $\alpha$ of the form $\sum_{k}a_{k}\left( \prod_{i=1}^e\alpha_i^{k_i} \right)$, where $k =(k_1,\ldots,k_e)$ ranges over $\integer_{[N_1]}\times\cdots\times\integer_{[N_e]}$, $(N_1,\ldots,N_e)\in\nat^{e}$, and $a_{k}\in\integer$. If $\alpha\neq0$ then $|\alpha|\geq \left(\sum_{k}|a_{k}| \right)^{1-h}\prod_{i=1}^{e}c^{-hN_i}$.
\end{lemma}

Since $\alpha_x$ is written in a polynomial form specified by Lemma \ref{lemma-Sto74}, this lemma provides us with an appropriate constant $c\in(0,1)$ satisfying that $\alpha_{x}\geq c^{|x|+1}$ for all $x\in\Sigma^*$ with $\alpha_x\neq0$.
For convenience, we assume that $c<2/3$.

In what follows, we shall construct a 2qfa $N$ that satisfies the desired condition of the corollary. This machine $N$ starts with the initial configuration $\psi_0 = \qubit{p_0,q_0}\qubit{0}$ and then transforms it to  $\psi_1 = \sqrt{c}\qubit{p_0,\bar{q}_0}\qubit{0} + \sqrt{1-c}\qubit{\bar{p}_0}\otimes U_{\delta}^{(x)}\qubit{q_0}\qubit{0}$.
Starting from the first term $\qubit{p_0,\bar{q}_0}\qubit{0}$ in $\psi_1$, for each index $h\in[0,n]_{\integer}$,  $N$ transforms $\qubit{p_0,\bar{q}_0}\qubit{h}$ into $\sqrt{c}\qubit{p_0,\bar{q}_0}\qubit{h+1} + \sqrt{\frac{1-c}{2}}\qubit{p_{acc,0}}\qubit{h} + \sqrt{\frac{1-c}{2}}\qubit{p_{rej,0}}\qubit{h}$. The last two terms respectively correspond to accepting and rejecting states, which are traced out immediately by measurements.
In addition, $N$ modifies $\qubit{p_0,\bar{q}_0}\qubit{n+1}$ to $\qubit{p_{acc,1}}\qubit{0}$.
The second term in $\psi_1$, in contrast, is composed of vectors in $\{\ket{\overline{p}_0}\ket{q}\ket{h}\mid q\in Q,h\in[0,n+1]_{\integer}\}$ and $N$ transforms each vector $\qubit{\bar{p}_0,q}\qubit{h}$ to $\qubit{\bar{p}_0}\otimes U_{\delta}^{(x)}\qubit{q}\qubit{h}$; however, we exchange between accepting states and rejecting states.

The acceptance probability $p_{N,acc}(x)$ of $N$ on $x$ is exactly $\frac{1}{2}(c-c^{n+2})+(1-c)p_{M,rej}(x)$ while the rejection probability $p_{N,rej}(x)$ equals $c^{n+2}+\frac{1}{2}(c-c^{n+2})+(1-c)p_{M,acc}(x)$.
When $x\in L$, since $p_{M,acc}(x)\geq \frac{1}{2}+c^{n+1}$ and $c<2/3$, we obtain
\[
p_{N,rej}(x)\geq (1-c)\left(\frac{1}{2} + c^{n+1}\right) +  \frac{1}{2}(c-c^{n+2}) = \frac{1}{2} + \left(1-\frac{3c}{2}\right)c^{n+1} >\frac{1}{2}.
\]
On the contrary, when $x\notin L$, since $p_{M,rej}(x)\geq \frac{1}{2}$, it follows that
\[
p_{N,acc}(x) \geq \frac{1}{2}(1-c) + \frac{1}{2}(c-c^{n+2}) = \frac{1}{2} + \frac{1}{2}c^{n+2}>\frac{1}{2}.
\]
Therefore, $L$ must belong to $\co\twopqfa_{\algebraic}(\abshalt)$.

(2) We first intend to prove the inclusion (*) $\twocequalqfa_{\algebraic}(\abshalt) \subseteq \co\twopqfa_{\algebraic}(\abshalt)$.
To show (*), let L be any language in $\twocequalqfa_{\algebraic}(\abshalt)$ and take a 2qfa $M=(Q,\Sigma,\delta,q_0,Q_{acc},Q_{rej})$ that
recognizes $L$ with exact cut point $1/2$.

We perform the following procedure on an arbitrary input $x$.
Starting with an initial state $\qubit{q_0}\qubit{q_0}$, we apply to it the operator $\overline{U}^{(x)}_{\overline{\delta}}= U^{(x)}_{\delta}\otimes U^{(x)}_{\delta}$. Accept $x$ if we reach $\qubit{q}\qubit{q'}$ for $(q,q')\in (Q_{acc}\times Q_{acc})\cup (Q_{rej}\times Q_{rej})$ and reject $x$ if $(q,q')\in (Q_{acc}\times Q_{rej})\cup (Q_{rej}\times Q_{acc})$.
Note that $x\in L$ implies $p_{N,acc}(x) = p_{M,acc}(x)^2+p_{M,rej}(x)^2 = 1/2$
and that $x\notin L$ leads to $p_{N,acc}(x) = 2p_{M,acc}(x)p_{M,rej}(x)<1/2$. by the definition of $\twopqfa$, $L$ belongs to $\co\twopqfa_{\algebraic}$.

The inclusion (*) implies $\co\twocequalqfa_{\algebraic}(\abshalt) \subseteq \twopqfa_{\algebraic}(\abshalt)$.  Since $\twopqfa_{\algebraic}(\abshalt)$ is closed under complementation by (1), we conclude from (*) that $\twocequalqfa_{\algebraic}(\abshalt) \subseteq \twopqfa_{\algebraic}(\abshalt)$, and thus the desired result follows.

(3) By Lemma \ref{folke-complex}, we obtain $\twoeqfa_{\complex}(\abshalt)=\twoeqfa_{\real}(\abshalt)$.
Since Theorem  \ref{abs-halt-vs-lin-time} yields the equality $\twoeqfa_{\real}(\abshalt) = \twoeqfa_{\real}[\lintime]$, it suffices to show that $\twoeqfa_{\real}[\lintime] \subseteq  \twoeqfa_{\algebraic\cap\real}(\abshalt)$.
Let $L$ be any language in $\twoeqfa_{\real}[\lintime]$ and let $M=(Q,\Sigma,\delta,q_0,Q_{acc},Q_{rej})$ be an error-free $\real$-amplitude 2qfa that recognizes $L$ in worst-case linear time.

Let $x\in\Sigma^*$ and $N=|Q|(|x|+2)$. As discussed in \cite[Section 6]{ADH97}, we consider \emph{symbolic computation} of $M$ on the input $x$ by replacing each transition amplitude of $M$ with a new variable. Assume that $\{\alpha_1,\alpha_2,\ldots,\alpha_m\}$ is a set of all real transition amplitudes used by $M$. For each amplitude $\alpha_i$, let $z_i$ denote a new variable associated with it. We set $z=(z_1,\ldots,z_m)$.
Corresponding to $U^{(x)}_{\delta}$, we define $U^{(x)}_{\delta}(z)$ to be a matrix obtained from $U^{(x)}_{\delta}$ by replacing each value $\alpha_i$ by $z_i$.  Since $U^{(x)}_{\delta}$ is unitary, we demand that $U^{(x)}_{\delta}(z)$ should be unitary as well.
Write $p_{M,acc}(x,z)$ and $p_{M,rej}(x,z)$ to denote the acceptance probability and the rejection probability produced by applying $U^{(x)}_{\delta}(z)$ in $N$ steps. By the linear time-bound of $M$, both $p_{M,acc}(x,z)$ and $p_{M,rej}(x,z)$ are expressed as polynomials in $z$.

Let $P_{i,j}^{(x)}(z)$ denote the \emph{dot product} of the $i$th and $j$th columns of $U_{\delta}^{x}(z)$.  Since $U^{(x)}_{\delta}(z)$ is required to be unitary, it must hold that $P_{i,i}^{(x)}(z)=1$ for all $i$'s and $P_{i,j}^{(x)}(z)=0$ for all distinct pairs $i,j$. Note that each $P^{(x)}_{i,j}$ is expressed as a certain polynomial in $z$. We then define $P_{M}(z)$ to be a set $\{P_{i,j}^{(x)}(z)\mid x\in\Sigma^*, i,j\in[1,N]_{\integer}, i\neq j \} \cup \{1-P_{i,i}^{(x)}(z)\mid i\in[1,N]_{\integer}, x\in\Sigma^*\}$. Finally, we  consider a set $I = P_{M}(z)\cup \{p_{M,acc}(x,z)\mid i\in\nat,x\in L\} \cup \{p_{M,rej}(x,z)\mid i\in\nat,x\in\overline{L}\}$ of polynomials in $z$.
Since $M$ produces no errors on all inputs, all polynomials in $I$ must have a common zero in $\real^m$.
We then apply the following result taken from \cite[Proposition 6.1]{ADH97}.

\begin{lemma}\label{ADH97-error-free}
Let $I$ be an ideal in $\rational[z_1,z_2,\ldots,z_m]$. If all polynomials in $I$ have a common zero in $\real^m$, then they also have a common zero in $(\algebraic\cap\real)^m$.
\end{lemma}

Lemma \ref{ADH97-error-free} guarantees that there should be solutions of all polynomials in $I$ within $(\algebraic\cap\real)^m$. By the definition of $I$, we can replace $P_{M}$ by a certain set of amplitudes in $\algebraic\cap\real$ without changing the outcomes of $M$ on all inputs and without altering the running time of $M$ on all inputs. This modification guarantees that $L$ is a member of $\twoeqfa_{\algebraic\cap\real}(\abshalt)$.
\end{proofof}

Hereafter, we shall discuss how to prove Lemma \ref{claim-Yao98}.
The core of the proof of this lemma is the \emph{Dimension Lemma} (Lemma \ref{Yao98-dimension}), which relates to the eventual behavior of each 2qfa, which performs a series of unitary operations and projective measurements.
This lemma is also an important ingredient in proving Lemma \ref{exp-time-bound} in Section \ref{sec:runtime-2qfa} and we thus need to zero in to the lemma.
To state this lemma, nonetheless, we first need to introduce a few notations.

Let $V=\complex^{N}$ denote an $N$-dimensional Hilbert space and let $U$ be any $N\times N$ unitary matrix over $V$. Moreover, let $W$ indicate a fixed nonempty subspace of $V$ and let $W^{\bot}$ be the dual space of $W$; that is, $V= W\oplus W^{\bot}$. We define $P_{W^{\bot}}$ to be the projection operator onto $W^{\bot}$. Obviously, $P_{W^{\bot}}(W)=\{0\}$ holds because $W\bot W^{\bot}$.  We then consider the operation $U_{W}=_{def} UP_{W^{\bot}}$. For convenience, we set $U^{0}_{W}(w)=w$ and define $U^{i+1}_{W}(w)= U_{W}(U^{i}_{W}(w))$ for any index $i\in\nat$ and any vector $w\in V$. Finally, we define $W_{i}=\{w\in V\mid U^{i+1}_{W}(w)=0\}$ for each $i\in\nat$ and we write  $W_{max}$ for $\bigcup_{i\in\nat}W_i$; in other words, $W_{max}=\{w\in V\mid \exists i\in\nat\;[U^{i+1}_{W}(w)=0]\}$.

\begin{lemma}\label{Yao98-dimension} \hs{1} {\rm [Dimension Lemma]}\hs{1}
There exists a number $d\in[0,N]_{\integer}$ for which $W_{max} = W_{d}$.
\end{lemma}

From this lemma, we can derive Lemma \ref{claim-Yao98} easily in the following fashion.

\vs{-2}
\begin{proofof}{Lemma \ref{claim-Yao98}}
Take any $\complex$-amplitude 2qfa $M=(Q,\Sigma,\delta,q_0,Q_{acc},Q_{rej})$ that halts absolutely. Given any index $n\in\nat$ and an arbitrary input $x\in\Sigma^n$, define $CONF_n= Q\times [0,n+1]_{\integer}$, and set $N=|CONF_n|$ and $V= \complex^{N}$. Recall three projection measurements $\Pi_{acc}$, $\Pi_{rej}$, and $\Pi_{non}$.

Take an arbitrary input $x\in\Sigma^*$ of length $n$. Let us consider a time-evolution matrix $U_{\delta}^{(x)}$ induced from $\delta$
and a halting  configuration space $W= span\{\qubit{q}\qubit{h} \mid q\in Q_{acc}\cup Q_{rej}, h\in[0,n+1]_{\integer}\}$ of $M$.  Since $P_{W^{\bot}} = \Pi_{non}$, it follows that  $U_{W}=U_{\delta}^{(x)} \Pi_{non}$.
Moreover, $W_i$ (defined above) can be expressed as $W_i = \{w\in V\mid (U_{\delta}^{(x)}\Pi_{non})^{i+1}(w)=0\}$. Write $w_0$ for $\qubit{q_0}\qubit{0}$. Since all computation paths of $M$ on the input $x$ terminate eventually, $w_0$ belongs to $W_{max}$. Lemma \ref{Yao98-dimension} implies that $W_{max}=W_{d}$ for a certain index $d\in[0,N]_{\integer}$. Thus, it follows that $w_0\in W_{d}\subseteq W_{N}$. This means that all the computation paths terminate within $N+1$ steps, as requested.
\end{proofof}

To close this section, we shall present the proof of Lemma \ref{Yao98-dimension}.

\begin{proofof}{Lemma \ref{Yao98-dimension}}
In this proof, we slightly modify Yao's original proof \cite{Yao98}.
Let $U$ be any $N\times N$ unitary matrix over $V=\complex^{N}$. Recall that the notation $W_i$ expresses $\{w\in V\mid U_{W}^{i+1}(w)=0\}$. First, we observe that $W_0=W$ because $U_{W}(w)=0$  iff $P_{W^{\bot}}(w)=0$ iff $w\in W$.
For each index $i\in\nat$, we set $K_{i+1}$ to be $U^{-1}(W_{i})\cap W^{\bot}$.

\begin{claim}\label{w-k}
For each index $i\in\nat$, $W_{i+1} = \mathrm{span}\{K_{i+1},W_i\}$.
\end{claim}

\begin{proof}
($\subseteq$) We want to show that $W_{i+1}\subseteq \mathrm{span}\{K_{i+1},W_i\}$ for any index $i\in\nat$.
Let $w$ be any vector in $W_{i+1}$ and express it as $x+y$ using two appropriate vectors $x\in W^{\bot}$ and $y\in W$. Since $P_{W^{\bot}}(x)=x$, we obtain $U^{i+2}_{W}(x) = U^{i+1}_{W}(UP_{W^{\bot}}(x)) = U^{i+1}_{W}(U(x))$. It thus follows that $U^{i+2}_{W}(w) = U^{i+2}_{W}(x+y) = U_{W}^{i+2}(x)+ U_{W}^{i+2}(y) = U^{i+1}_{W}(U(x))$ because of $U_{W}(y)=0$.  Moreover, $w\in W_{i+1}$ implies $U_W^{i+2}(w)=0$. {}From this result, we conclude that $U^{i+1}_{W}(U(x))=0$. This implies $U(x)\in W_{i}$; in other words, $x\in U^{-1}(W_i)$. Since $x\in W^{\bot}$, $x$ must belong to $(U^{-1}(W_i))\cap W^{\bot}$, which equals $K_{i+1}$. From the facts $w=x+y$ and $y\in W\subseteq W_{i+1}$, it follows that $w$ is in $\mathrm{span}\{K_{i+1},W_{i+1}\}$.

($\supseteq$) Next, we wish to prove that $\mathrm{span}\{K_{i+1},W_i\} \subseteq W_{i+1}$ for any index $i\in\nat$.
Let $w$ be of the form $x+y$ for certain vectors $x\in K_{i+1}$ and $y\in W_i$. Since $y\in W_i$, we obtain $U^{i+1}_{W}(y)=0$. By the definition of $K_{i+1}$, $x$ belongs to both $W^{\bot}$ and $U^{-1}(W_i)$. For simplicity, we set $z=U(x)$. Since $x\in U^{-1}(W_i)$, $z$ must be in $W_i$. From this follows $U^{i+1}_{W}(z)=0$, which implies $U^{i+1}_{W}(U(x))=0$. Note that $P_{W^{\bot}}(x)=x$ because $x\in W^{\bot}$. It thus follows that $U^{i+2}_{W}(x) = U^{i+1}_{W}(UP_{W^{\bot}}(x)) = U^{i+1}_{W}(U(x)) =0$. Therefore, we obtain $U^{i+2}_{W}(w) = U^{i+2}_{W}(x+y) = U^{i+2}_{W}(x) + U^{i+2}_{W}(y) = 0$, which obviously indicates that $w\in W_{i+1}$ by the definition of $W_{i+1}$.
\end{proof}

We note by the definition of $W_i$ that the inclusion $W_{i}\subseteq W_{i+1}$ holds for every index $i\in\nat$. Claim \ref{w-k} therefore yields the following equivalence relation.

\begin{claim}\label{inclusion}
For any number $i\in\nat$, $K_{i+1}\subseteq W_i$ iff $W_{i+1} = W_i$.
\end{claim}

\begin{proof}
If $K_{i+1}\subseteq W_i$, then $\mathrm{span}\{K_{i+1},W_i\} = W_i$.
Claim \ref{w-k} thus implies that $W_{i+1} = \mathrm{span}\{K_{i+1},W_i\} = W_i$. Conversely, if $W_{i+1} = W_i$, then we use Claim \ref{w-k} and obtain  $K_{i+1}\subseteq \mathrm{span}\{K_{i+1},W_i\} = W_{i+1} = W_i$.
\end{proof}

Let $d$ denote the minimal natural number $i$ satisfying $W_i=W_{i+1}$. Such a number exists because $W$ is a finite-dimensional space.

\begin{claim}\label{w-max}
Let $i$ be any number in $\nat$.
If $i<d$, then $W_{i}\subsetneq W_{i+1}$; otherwise, $W_i = W_{i+1}$.
\end{claim}

\begin{proof}
Clearly, $W_i\subseteq W_{i+1}$ holds for any index $i\in\nat$. The first part of the claim is trivial because $d$ is the minimal number satisfying $W_i=W_{i+1}$. The second part of the claim can be proven by induction on $i\geq d$. The basis case $W_{d}=W_{d+1}$ is true because of the definition of $d$. Take any index $i>d$ and assume that $W_{i}=W_{i+1}$ holds. This assumption is equivalent to $K_{i+1}\subseteq W_i$ by Claim \ref{inclusion}.
We want to verify that $K_{i+2}\subseteq W_{i+1}$. For this assertion, from $W_{i}=W_{i+1}$, we derive $K_{i+2} = U^{-1}(W_{i+1})\cap W^{\bot} = U^{-1}(W_i)\cap W^{\bot} = K_{i+1}$. Since $K_{i+1}\subseteq W_{i}$ by the induction hypothesis,  we immediately obtain $K_{i+2}\subseteq W_{i}$, which implies $W_{i+1}=W_{i+2}$ by Claim \ref{inclusion}.
\end{proof}

Claim \ref{w-max} implies that $W_d=W_i$ for any index $i\geq d$. Hence, we obtain $W_d=W_{max}$. How large is this $d$? Note that $dim(W_0)\geq1$ since $W=W_{0}$ and $W$ is nonempty. It thus follows by Claim \ref{w-max} that $dim(W_i)<dim(W_{i+1})$ for any $i<d$ and $dim(W_i)=dim(W_d)$ for any $i\geq d$. Therefore, we conclude that $d\leq dim(W_d)\leq N$.

This completes the proof of Lemma \ref{Yao98-dimension}.
\end{proofof}

\subsection{Running-Time Bounds of QFAs}\label{sec:runtime-2qfa}

We have given in Section \ref{sec:absolute-halt} a linear upper bound of the running time of  absolutely-halting 2qfa's. In general, not all computation paths of bounded-error 2qfa's may terminate. Even though, we can claim that it is sufficient to focus  only on their computation paths that actually terminate {\em within   exponential time} and to ignore all the other computation paths in order to recognize languages with bounded-error probability.

To state this claim formally, we need to define a restricted form of 2qfa's. Here, we shall treat any computation path that does not enter a halting state within $t(n)$ steps as ``unhalting'' and such a computation path is conveniently categorized as neither accepting nor rejecting. More precisely,
a {\em $t(n)$ time-bounded} 2qfa $M$ is a variant of 2qfa that satisfies the following condition: we force $M$ to stop applying its transition matrix after exactly $t(n)$ steps (unless it halts earlier) and, after this point, we ignore any computation step taken along any computation path by viewing such a computation path as ``unhalting.''

\begin{theorem}\label{complete-halt-time-bound}
Any language in $\twobqfa_{\algebraic}$ can be  recognized by a certain $2^{O(n)}$ time-bounded 2qfa with bounded-error probability.
\end{theorem}

Unfortunately, we cannot expand the scope of Theorem \ref{complete-halt-time-bound} to $\complex$-amplitude 2qfa's or unbounded-error 2qfa's because the theorem is derived from the following lemma, which heavily relies on Lemma \ref{lemma-Sto74}.

\begin{lemma}\label{exp-time-bound}
Let $M$ be any $\algebraic$-amplitude 2qfa with a set $Q$ of inner states with error probability at most $\varepsilon$, where $\varepsilon\in[0,1/2]$. Let $\varepsilon'=(1-2\varepsilon)/4$.
There exist a constant $c>0$ and a $c^{|Q|(n+2)}$ time-bounded 2qfa $N$  that satisfy the following: for any input $x$, (i) $M$ accepts (resp., rejects) $x$ with probability at least $1-\varepsilon$ iff $N$ accepts (resp., rejects) $x$ with probability at least $1-\varepsilon'$.
\end{lemma}

Here, let us derive Theorem \ref{complete-halt-time-bound} from Lemma \ref{exp-time-bound}.

\vs{-2}
\begin{proofof}{Theorem \ref{complete-halt-time-bound}}
Let us consider a language $L$ in $\twobqfa_{\algebraic}$ and take an $\algebraic$-amplitude 2qfa $M$ that recognizes $L$ with error probability at most $\varepsilon\in[0,1/2)$. Lemma \ref{exp-time-bound} provides a constant $c>0$ and another 2qfa, say, $N$ such that (i) $N$ is $c^{|Q|(n+2)}$ time-bounded, where $Q$ is a set of $M$'s inner states, and (ii) $p_{M,e}(x)\geq 1-\varepsilon$ iff $p_{N,e}(x)\geq 1-\varepsilon'$, for each type $e\in\{acc,rej\}$ and for every input $x$. This implies that, since $\varepsilon'\in[0,1/2)$, $L$ can be recognized by $N$ with bounded-error probability.
\end{proofof}

To prove Theorem \ref{complete-halt-time-bound}, we need to verify the correctness of Lemma \ref{exp-time-bound}. In the following proof of the lemma, we shall stick to the same terminology introduced in Section \ref{sec:absolute-halt}.
An underlying idea of the proof is to show how to estimate the running time of a given 2qfa by evaluating  \emph{eigenvalues} of its time-evolution matrix.

\vs{-2}
\begin{proofof}{Lemma \ref{exp-time-bound}}
First, take any 2qfa $M=(Q,\Sigma,\delta,q_0,Q_{acc},Q_{rej})$ with $\algebraic$-amplitudes with error probability at most $\varepsilon\in[0,1/2]$.
Fix $n$, the length of inputs, and let $N=|Q|(n+2)$, the total number of configurations of $M$ on inputs of length $n$.
For simplicity, let $V=\complex^{N}$ be the configuration space of $M$ on inputs of length $n$.
Hereafter, we arbitrarily fix $x$ in $\Sigma^n$ and write $U$ for a unique transition matrix $U_{\delta}^{(x)}$ that dictates a single move of $M$ on the input $x$. Let us recall three notations $W_{acc}$, $W_{rej}$, and $W_{non}$ from Section \ref{sec:quantum-finite-automata}.
By setting  $W=W_{acc}\oplus W_{rej}$ and $W^{\bot}=W_{non}$, we obtain $U_{W}$, $W_{i}$, and $W_{max}$ as in Section \ref{sec:absolute-halt}.
We assume that the initial inner state $q_0$ is a \emph{non-halting state}, because, if $q_0$ is a halting state, the lemma is trivially true because $M$ is already $1$ time-bounded. In what follows, we assume that the initial superposition of $M$ is in $W^{\bot}$. By Lemma \ref{Yao98-dimension}, there exists a number $d'\in[0,N]_{\integer}$ such that $W_{max}=W_{d'}$; in other words, any element $v\in W_{max}$ is mapped into $W$ within $d'+1$ steps.
For simplicity, we set $\tilde{U}_{W} = U_{W}^{d'+1}$.

Here, we assume that $\dim(W_{max})<N$ and let $m$ denote the dimension of $W_{max}^{\bot}$. Without loss of generality, we assume that any vector $v$ in $W_{max}^{\bot}$ can be expressed as an $N$ dimensional column vector of the form $v=(w,0,\ldots,0)^{T}$, where $w$ has $m$ entries of the form  $(w_1,w_2,\ldots,w_m)$. This assumption helps us express any superposition of configurations of $M$ as a vector $v$ of the form $(w_1,w_2,\ldots,w_{N})^T$ in $V$, where the last $N-m$ entries ``correspond'' to $W_{max}$.
Hence, $\tilde{U}_{W}$ is written as
$
\tilde{U}_{W} =
\begin{pmatrix}
A & O \\
B & O
\end{pmatrix},
$
where $A$ is an $m\times m$ matrix and $B$ is an $(N-m)\times m$ matrix (namely, $A$ is a linear map from $\complex^{m}$ to $\complex^{m}$ and  $B$ is a linear map from $\complex^{m}$ to $\complex^{N-m}$).

Given a vector $v=(w,0,\ldots,0)^{T}\in W_{max}^{\bot}$, we obtain $\tilde{U}_{W}(v) = (Aw,Bw)^{T}$. For each index $k\in\nat^{+}$, it follows that  $\tilde{U}^{k}_{W}(v) = (A^kw,BA^{k}w)^{T}$. Note that $(A^kw,0,\ldots,0)^T\in W_{max}^{\bot}$ and  $(0,\ldots,0,BA^kw)^{T}\in W_{max}$. Thus, $\tilde{U}_{W}$ must map $(0,\cdots,0,BA^kw)^T$ into $W$.
In other words, $(0,\ldots,0,BA^kw)^T$ is mapped by $M$ into $W$ within
$N+1$ steps.

Next, we shall argue that $A$ is diagonalizable in $\complex$.
Let $nullity(A)$ denote the dimension of the null space $Null(A)=\{w\in\complex^{m}\mid (w,0,\ldots,0)^T\in W_{max}^{\bot},  Aw=0\}$.
For our purpose, we intend to verify the equality $nullity(A)=0$, which is essentially equivalent to $Null(A)=\{0\}$ by way of contradiction. Toward a contradiction, assume that there is a non-zero element $w$ in $Null(A)$.
For the vector $v=(w,0,\ldots,0)^T$, since $Aw=0$, we obtain $\tilde{U}_{W}(v)=(0,\ldots,0,Bw)\in W_{max}$, which implies that $v\in W_{max}$, a contradiction against $v\in W_{max}^{\bot}$.
Therefore, we conclude that $Null(A)=\{0\}$. Since $rank(A) + nullity(A) = dim(\complex^{m}) =m$, the rank of $A$ equals $m$. This means that $A$ has its inverse $A^{-1}$ and, consequently, $A$ is diagonalizable in $\complex$.

Let $\{\lambda_1,\ldots,\lambda_m\}$ denote a set of all eigenvalues of $A$ and let $\{v_1,\ldots,v_m\}$ be a set of their associated \emph{unit-length}  eigenvectors (i.e., $\|v_i\|=1$ for any index $i\in[m]$). For convenience, we assume that those eigenvalues are sorted in increasing order according to their absolute values. Take the maximal index $i_0$ such that $|\lambda_i|<1$ holds for all $i\leq i_0$ and $|\lambda_i|=1$ for all $i>i_0$. Since $A$ is diagonalizable in $\complex$, find an appropriate unitary matrix $P$ satisfying
\[
A = P^{\dagger}
\begin{pmatrix}
\lambda_1 & & & \\
 & \lambda_2 &  & O\\
 O & & \ddots & \\
 & & & \lambda_m
 \end{pmatrix}
 P.
\]
Let the \emph{undetermined space}  $D_{und}$ be $\mathrm{span}\{v_1,v_2,\ldots,v_{i_0}\}$ and let the \emph{stationary space}  $D_{sta}$ be $\mathrm{span}\{v_{i_0+1},v_{i_0+2},\ldots,v_{m}\}$. Obviously, $\complex^{m} = D_{und}\oplus D_{sta}$ holds. Note that  if $w\in D_{sta}$ then $\|Aw\| = \|w\|$, implying $Bw=0$. This means that, once $w$ falls into $D_{sta}$, $\tilde{U}_W((w,0,\ldots,0)^T)$ is also in $D_{sta}\otimes\{0\}^{m-i_0}$.
In contrast, when $w\in D_{und}$, since $w$ is of the form $\sum_{1\leq j\leq i_0}\alpha_j v_j$ for certain coefficients $\alpha_1,\ldots,\alpha_{i_0}$, it follows that $Aw = \sum_{j}\alpha_j\lambda_j v_j$. Define  $\lambda_{max}=\max_{1\leq j\leq i_0}\{\lambda_j\}$.
Since $\lambda_i$'s are sorted, we obtain   $|\lambda_{max}|<1$, from which we conclude that $\|Aw\|^2 \leq |\lambda_{max}|^2\sum_{j}|\alpha_j|^2 = |\lambda_{max}|^2\|w\|^2$ (since $\|v_j\|=1$). Therefore, $\|Aw\|\leq |\lambda_{max}|\|w\|$ follows. This fact implies that $\|A^kw\|\leq |\lambda_{max}|^k \|w\|$ for any number $k\geq1$. From this inequality, it follows that $\lim_{k\rightarrow \infty}\|A^kw\| \leq \lim_{k\rightarrow \infty}|\lambda_{max}|^k\|w\| =0$.

For notational convenience, let $D_{sta}^* = D_{sta}\otimes\{0\}^{N-m}$ and $D_{und}^*=D_{und}\otimes\{0\}^{N-m}$. Note that $W_{max}^{\bot}
= D_{und}^*\oplus D_{sta}^*$. Next, we define $D=D_{sta}^*\oplus W_{max}$. It follows that $D^{\bot}=D_{und}^*$ since $V=W_{max}^{\bot}\oplus W_{max}$. Now, let $P_{D^{\bot}}$ express a unique projection onto $D^{\bot}$ and let $U_{D}$ be the operation $UP_{D^{\bot}}$. Define $U^{0}_{D}(w)=w$ and $U^{i+1}_{D}(w) = U_{D}(U^{i}_{D}(w))$ for each $i\in\nat$. Finally, define $D_i=\{w\in V\mid U^{i+1}_{D}(w)=0\}$ for every $i\in\nat$ and let $D_{max}=\{w\in V\mid \exists i\in\nat\,[U^{i+1}_{D}(w)=0]\}$. By the Dimension Lemma (Lemma \ref{Yao98-dimension}),we choose an appropriate index $d\in[0,m]_{\integer}$ for which  $D_{max}=D_{d}$ holds.

To make the rest of this proof simple, we rearrange the coordinate system for $V$ to match the order of $\{v_1,v_2,\ldots,v_m\}$.  We modify $M$ to define $\tilde{M}$ so that $\tilde{M}$ applies $\tilde{U}_{W}$ (instead of $U_{W}$) in a single step. For the sake of simplicity, let
$\tilde{A} =
\begin{pmatrix}
A & O \\
O & O
\end{pmatrix}$ and $\tilde{B} =
\begin{pmatrix}
O & O \\
B & O
\end{pmatrix}$.
Let us determine a series of vectors $w_0,w_1,w_2,\ldots$, which are generated by running $\tilde{M}$ starting with its initial configuration $w_0$.

Let $w_0$ express  the initial configuration $\qubit{q_0}\qubit{0}$ of $\tilde{M}$ on $x$. Since $V = D_{und}^*\oplus D_{sta}^*\oplus W_{max}$, $w_0$ can be expressed as $x_0+y_0+z_0$ for appropriate vectors $x_0\in D_{und}^*$, $y_0\in D_{sta}^*$, and $z_0\in W_{max}$.
We note that $x_0+y_0\in W_{max}^{\bot}$ and $\|w_0\|^2 = \|x_0\|^2 + \|y_0\|^2+\|z_0\|^2$.
Here, we do not need to consider $y_0$ or $z_0$ because $y_0$ will not terminate  and $z_0$ will terminate at the next step of $\tilde{M}$; it thus suffices to pay our attention to $x_0$.
At the next step of $\tilde{M}$, we apply $\tilde{U}_{W}$ to $x_0$. Let $w_1= \tilde{U}_{W}(x_0)$. This vector $w_1$ equals $\tilde{A}x_0 + \tilde{B}x_0$ and is written as $x_1 + y_1 + z_1$, where  $\tilde{A}x_0 = x_1 + y_1$  and $\tilde{B}x_0=z_1$ for certain vectors $x_1\in D_{und}^*$ and $y_1\in D_{sta}^*$. For the same reason as before, we must zero in only to $x_1$. More generally, at Step $i$, we obtain $w_i = \tilde{U}_{W}(x_{i-1})$ and $w_i$ must have the form $x_i+y_i+z_i$ for three vectors $x_i\in D_{und}^*$, $y_i\in D_{sta}^*$, and $z_i=\tilde{B}x_{i-1}$ satisfying $\tilde{A}x_{i-1}=x_i+y_i$.

Since all $x_i$'s are in $D_{und}^*$ ($=D^{\bot}$), the above process of generating $w_i$ from $x_{i-1}$ is the same as applying $U_{D}$ ($=UP_{D^{\bot}}$) to $x_{i-1}$; that is, $w_i= U_{D}(x_{i-1}) = x_i+ (y_i+z_i)$ for all $i\geq1$. Since $D_{max}=D_{d}$, we obtain $U_{D}(x_{d})=0$. Since $d\leq m \leq N$, the above computation of $\tilde{M}$ must end. In terms of $M$, each $z_i$ requires at most $N+1$ steps of $M$ in order to be mapped to $0$.

Let $\varepsilon' = \frac{1}{2}(\frac{1}{2}-\varepsilon)>0$. Note that $\varepsilon'$ is  a constant because so is $\varepsilon$.
Since our 2qfa $M$ halts with probability at least $1-\varepsilon$,
it must hold that, for any sufficiently large natural number $k$,
$\|\tilde{A}^kv\|^2 = \|A^kw\|^2 \leq (|\lambda_{max}|^{k}\|w\|)^2  \leq |\lambda_{max}|^{2k} \leq \varepsilon'$
for all vectors $v=(w,0,\ldots,0)^T\in W_{max}^{\bot}$ with $w\in\complex^m$. The last inequality implies that $k\leq (\log{\varepsilon'})/(2\log{|\lambda_{max}|})$.

Here, we need to find a polynomial upper-bound of the value $|\lambda_{max}|$. For this purpose, let $\alpha = 1-|\lambda_{max}|$ and let $T$ denote the set of all amplitudes used by $M$.  To apply Lemma \ref{lemma-Sto74}, we want to assert that $\alpha$ can be expressed as a certain form of polynomial. In  the case of quantum Turing machines, we refer the reader to \cite{Yam03}.
Recall that our amplitudes are all drawn from $\algebraic$.  Let $S=\{\alpha_1,\ldots,\alpha_e\}$ denote the maximal subset of $T$ that is algebraically independent, where  $e$ ($\in\nat^{+}$) satisfies $e\leq |Q||\Sigma||D|$. Define $F=\rational(S)$ and let $G$ be a field generated by all elements in $\{1\}\cup(T-S)$ over $F$. We write  $\{\beta_0,\beta_1,\ldots,\beta_{h-1}\}$ for a basis of $G$ over $F$ with $\beta_0=1$ and define $T'=T\cup\{\beta_i\beta_j\mid i,j\in\integer_h\}$. Take any common denominator $u$ such that, for every $\gamma\in T'$, $u\gamma$ is of the form $\sum_{t}a_{t}\left( \prod_{i=1}^{e}\alpha_{i}^{t_i} \right)\beta_{t_0}$, where $t=(t_0,t_1,\ldots,t_e)$ ranges over $\integer_h\times\integer^m$ and $a_{t}\in\integer$.
It is possible to choose a number $a\in\nat^{+}$ for which the amplitude
of any configuration of $M$ at time $k$ on $x$, when multiplied by $u^{2k-1}$, must have the form $\sum_{t}a_{t}\left(\prod_{i=1}^{e}\alpha_i^{t_i}  \right)\beta_{t_0}$, where $t =(t_0,t_1,\ldots,t_e)$ ranges over $\integer_h\times\left( \integer_{[2ak]}\right)^e$ and $a_{t}\in\integer$. Therefore, $\alpha$ is written in a polynomial form.

Since $\alpha\neq0$, we conclude by Lemma \ref{lemma-Sto74} that, for an appropriate choice of $c>0$, $|\alpha|\geq c^{-N}$ holds; in other words,  $1-|\lambda_{max}|\geq c^{-N}$ or equivalently $|\lambda_{max}|\leq 1- c^{-N}$ holds. The last inequality implies that $\log{|\lambda_{max}|^{-1}} \geq \log(1-c^{-N})^{-1} \geq c^{-N}$.
It therefore follows that $k\leq (\log{\varepsilon'})/(2\log{|\lambda_{max}|}) = (\log{(\varepsilon')^{-1}})/(2\log{|\lambda_{max}|^{-1}}) \leq c'c^{N}$ for another appropriate constant $c'>0$. This obviously yields the lemma.
\end{proofof}

\section{Non-Recursive Languages}\label{sec:non-recursive}

The use of \emph{unrestricted amplitudes} often endows underlying qfa's with enormous computational power, and consequently it causes the qfa's to recognize even non-recursive languages. In what follows, we wish to discuss what type of 2qfa's recognizes non-recursive languages when arbitrary amplitudes are allowed. In our study, however, we shall pay our attention only to qfa's that halt absolutely with various accepting criteria.

We start with a simple claim that all languages in $\twoeqfa_{\complex}(\abshalt)$ are recursive even if all amplitudes used by underlying 2qfa's are not recursive. For notational convenience, we write $\mathrm{REC}$ for the family of all {\em recursive languages}.

\begin{proposition}\label{EQFA-recursive}
$\twoeqfa_{\complex}(\abshalt) \subsetneqq \mathrm{REC}$.
\end{proposition}

\begin{proof}
From  Corollary \ref{amplitude-reduction}(3) and Theorem \ref{abs-halt-vs-lin-time}, we obtain  $\twoeqfa_{\complex}(\abshalt) = \twoeqfa_{\algebraic\cap\real}(\abshalt) = \twoeqfa_{\algebraic\cap\real}[\lintime]$.  It thus suffices to verify that (*) $\twoeqfa_{\algebraic\cap\real}[\lintime]\subsetneqq \mathrm{REC}$.

Adleman \etalc~\cite{ADH97} demonstrated that the language family $\mathrm{EQP}_{\complex}$ (error-free quantum polynomial time) with $\complex$-amplitudes is contained in $\mathrm{REC}$. It is rather clear that $\twoeqfa_{\algebraic\cap\real}[\lintime]\subseteq \mathrm{EQP}_{\complex}$. Since $\mathrm{EQP}_{\complex}\subseteq \mathrm{REC}$, we conclude that $\twoeqfa_{\algebraic\cap\real}[\lintime]\subseteq \mathrm{REC}$.

For the separation in (*), it suffices to construct a recursive language that is recognized by no $\algebraic\cap\real$-amplitude error-free 2qfa's $M$ running in worst-case $|Q|(n+2)+1$ time by Lemma \ref{claim-Yao98}. This task can be done by a standard diagonalization argument. First, we encode each $\algebraic\cap\real$-amplitude 2qfa into a certain binary string by treating amplitudes using their defining polynomials. Next, we enumerate the encodings of all 2qfa's and define $L$ to be a set of all strings $x$ such that $x$ encodes a certain $\algebraic\cap\real$-amplitude 2qfa $M=(Q,\{0,1\},\delta,q_0,Q_{acc},Q_{rej})$ and $M$ does not accept $x$ within $|Q|(|x|+2)+1$ steps. By the definition of $L$, $L$ does not belong to $\twoeqfa_{\algebraic\cap\real}[\lintime]$.
Since $L$ is recursive by its recursive construction, the desired separation follows instantly.
\end{proof}

The case of unbounded-error probability is quite different from Proposition \ref{EQFA-recursive}. Since $\mathrm{SL}_{\real}$ is known to be uncountable \cite{Rab63}, Lemma \ref{PQFA-equal-SL} immediately implies that $\onepqfa_{\complex}$
contains a non-recursive language. Corollary \ref{onePQFA-included-2PQFA} then helps us conclude that the same is true for $\twopqfa_{\complex}[\lintime]$. For completeness, we include the entire proof of this result using our terminology.

\begin{proposition}\label{QFA-non-recursive-NO1}
$\twopqfa_{\complex}[\lintime]\nsubseteq \mathrm{REC}$. More strongly,  $\onepqfa_{\complex}\nsubseteq \mathrm{REC}$.
\end{proposition}

\begin{proof}
In the classical case, Rabin \cite{Rab63} implicitly argued that $\mathrm{SL}_{\real}$ is uncountable; thus, it must contain a non-recursive language.  The proposition comes directly from his claim. For completeness, we include the proof of the claim.

\begin{claim}\label{SL-nonrecursive}
$\mathrm{SL}_{\real}$  is uncountable.
\end{claim}

\begin{proof}
Let $\Sigma=\{0,1\}$. Since we are allowed to take any real cut point, we choose an arbitrary real number $\varepsilon\in(0,1]$. Using this $\varepsilon$, we define $L_{\varepsilon} =\{x^R\in\Sigma^*\mid 0.x>\varepsilon\}$, where $0.x$ is the binary expansion of each real number in $[0,1)$. Note that there are uncountably many such languages $L_{\varepsilon}$. Here, we want to show that $L_{\varepsilon}\in \mathrm{SL}_{\real}$ by constructing a 1pfa $M = (Q,\Sigma,\delta,q_0,Q_{acc},Q_{rej})$ such that $p_{M,acc}(x^R) = 0.x/2\varepsilon>1/2$ iff $0.x>\varepsilon$. Define $Q=\{q_0,q_1,q_2,q_3,q_4\}$, $Q_{acc}=\{q_3\}$, and $Q_{rej}=\{q_4\}$. Our transition function $\delta$ is described by the following transition matrices: for any symbol $\alpha\in\{\cent,0,1\}$,
\[
P_{\alpha} = \lmatrices{A_{\alpha}}{O_{3\times2}}{O_{2\times3}}{I_{2\times2}} \;\; \text{and} \;\;
P_{\dollar} = \lmatrices{O_{3\times3}}{O_{3\times2}}{B}{I_{2\times2}},
\]
where $O_{k\times l}$ and $I_{k\times l}$ are respectively the $k\times l$ zero matrix and the $k\times l$ identity matrix, and
\[
A_{\cent} = \ninematrices{\frac{1}{2\varepsilon}}{0}{0}{0}{0}{0}{\frac{2\varepsilon-1}{2\varepsilon}}{1}{1}, \;\; A_{0} = \ninematrices{1}{\frac{1}{2}}{0}{0}{\frac{1}{2}}{0}{0}{0}{1},
\;\; A_{1} = \ninematrices{\frac{1}{2}}{0}{0}{\frac{1}{2}}{1}{0}{0}{0}{1}, \text{ and }
\; B = \twothreematrices{0}{1}{0}{1}{0}{1}.
 \]
By a direct calculation, we obtain $(0,0,0,1,0)P_{\dollar}P_{x_n}P_{x_{n-1}}\cdots P_{x_1}P_{\cent}(1,0,0,0,0)^T = 0.x_{n}x_{n-1}\cdots x_1$, implying $p_{M,acc}(x^R) = \frac{0.x}{2\varepsilon}$. Thus, it follows that  $p_{M,acc}(x^R)>\frac{1}{2}$ iff $0.x>\varepsilon$ iff  $x^R\in L_{\varepsilon}$. We therefore conclude that $L_{\varepsilon}$ is in $\mathrm{SL}_{\real}$.
\end{proof}

From Claim \ref{SL-nonrecursive}, we immediately conclude that  $\mathrm{SL}_{\real}$ contains a non-recursive language.
Since $\onepqfa_{\complex}\subseteq \twopqfa_{\complex}[\lintime]$ by Corollary \ref{onePQFA-included-2PQFA}, Lemma \ref{PQFA-equal-SL} implies that $\mathrm{SL}_{\real}\subseteq \twopqfa_{\complex}[\lintime]$. Therefore,  $\twopqfa_{\complex}[\lintime]$ also has a non-recursive language.
\end{proof}

In the bounded-error case, $\twobqfa_{\complex}(\abshalt)$ is situated between $\twoeqfa_{\complex}(\abshalt)$ and $\twopqfa_{\complex}(\abshalt)$. It thus natural to ask whether $\twobqfa_{\complex}(\abshalt)$ is large enough to contain a  non-recursive language. Unfortunately, we cannot answer this question; instead, we make a slightly weak claim, Proposition
\ref{QFA-non-recursive-NO2}.

\begin{proposition}\label{QFA-non-recursive-NO2}
$\twobqfa_{\complex}(2\mbox{-}head,\abshalt)\nsubseteq \mathrm{REC}$.
\end{proposition}

\begin{proof}
For this proposition, we want to verify that
$\twobqfa_{\complex}(2\mbox{-}head,\abshalt)$ is uncountable by proving that it contains all subsets of $\{a^n\mid \exists m\in\nat\,[n=3^{2m}]\}$ since there are uncountably many such subsets.
The following argument comes from the proof of \cite[Theorem 5.1]{ADH97}, which demonstrates that $\mathrm{BQP}_{\complex}\nsubseteq\mathrm{REC}$.
We start with defining $A_3$ as the
set $\{a^n\mid \exists m\in\nat\,[n=3^{2m}]\}$ and we take any language $L$ over unary alphabet $\Sigma=\{a\}$ satisfying  $L\subseteq A_{3}$.
In what follows, we conveniently use $L$ to express its \emph{characteristic function} (\ie $L(x)=1$ for all $x\in L$ and $L(x)=0$ for all $x\notin L$).
Associated with $L$, we define a real number $\theta = 2\pi \sum_{n\in A_3} (-1)^{L(a^{n})} \cdot \frac{1}{9n}$. We first claim that $A_3$ can be recognized by a certain 2head-2rfa; that is, $A_3\in\tworfa(2\mbox{-}head)$.

\begin{claim}\label{2head-2rfa-forA3}
There exists a 2head-2rfa $M$ that recognizes the language $A_3$ over the alphabet $\Sigma=\{a\}$.
\end{claim}

\begin{proof}
We shall construct a 2head-2rfa $M=(Q,\Sigma,\delta,q_0,Q_{acc},Q_{rej})$ that recognizes $A_3$. Let $Q=\{ (q_i,q_k)\mid i\in[0,3]_{\integer},k\in\{odd,even,acc\}\}$, $Q_{rej}=\{ (q_i,q_{rej})\mid i\in[0,3]_{\integer}\}\}$, $Q_{acc}=\{ (q_3,q_{acc})\}$, and $D=\{0,\pm1\}$.
Our transition function $\delta: Q\times \check{\Sigma}\times\check{\Sigma}\to Q\times D\times D$ instructs $M$ to behave as follows.
In each odd round $2j+1$ for $j\geq0$, $M$ moves the second tape head forward $3$ cells  while  the first tape head stays stationary for $3$ steps and moves back for $1$ step. Whenever one of the tape heads returns to $\cent$, $M$ switches the roles of $2$ tape heads   until at least one tape head reaches $\dollar$. More precisely, for each $i\in[3]$, an inner state $(q_i,q_{odd})$ (resp., $(q_i,q_{even})$) indicates that $M$ has already elapsed for $i$ steps in round  $2j+1$ (resp., $2j+2$) for a certain number $j\geq0$.

\begin{table}[htbp] \label{table-transition}
\begin{tabular}{|l|c|c|l|r|r|} \hline
inner state &  symb. 1 & symb. 2 & inner state & dir. 1 & dir. 2 \\ \hline
$(q_0,q_0)$ & $\cent$ & $\cent$ & $(q_1,q_{odd})$ & $+1$ & $+1$ \\
$(q_1,q_{odd})$ & $\dollar$ & $\dollar$ & $(q_{1},q_{rej})$ & $-1$ & $-1$ \\
$(q_i,q_{odd})$, $i=1,2$ & $a$ & $a$ & $(q_{i+1},q_{odd})$ & $0$ & $+1$ \\
$(q_3,q_{odd})$ & $a$ & $a$ & $(q_{1},q_{odd})$ & $-1$ & $0$ \\
$(q_{3},q_{odd})$ & $\cent$ & $a$ & $(q_{1},q_{even})$ & $+1$ & $0$ \\
$(q_{i},q_{odd})$, $i\in[3]$  & $a$ & $\dollar$ & $(q_{i},q_{rej})$ & $0$ & $-1$ \\
$(q_i,q_{even})$, $i=1,2$ & $a$ & $a$ & $(q_{i+1},q_{even})$ & $+1$ & $0$ \\
$(q_3,q_{even})$ & $a$ & $a$ & $(q_{1},q_{even})$ & $0$ & $-1$ \\
$(q_{3},q_{even})$ & $a$ & $\cent$ & $(q_{1},q_{odd})$ & $0$ & $+1$ \\
$(q_3,q_{even})$  & $\dollar$ & $a$ & $(q_{3},q_{acc})$ & $-1$ & $0$ \\
$(q_i,q_{even})$, $i\neq3$ & $\dollar$ & $a$ & $(q_{i},q_{rej})$ & $-1$ & $0$ \\
\hline
\end{tabular}
\caption{Transitions of $M$.  The first row, for example, indicates $\delta((q_0,q_0),\cent,\cent) = ((q_1,q_{odd}),+1,+1)$.}
\end{table}

Table 1 formally describes $\delta$.
Note that, after round $2j$ (resp., $2j+1$), the first (resp., second) tape head must have moved to the cell indexed $3^{2j}$ (resp., $3^{2j+1}$). Therefore, we conclude that $x$ is in $A_3$ iff the first tape head reaches $\dollar$ in a unique accepting state $(q_3,q_{acc})$.
\end{proof}

To recognize $L$ using a 2head-2qfa, say, $N$, it suffices for us to implement  the following procedure on $N$. Let $R_{\theta}$ be a \emph{rotation matrix} $\lmatrices{\cos\theta}{-\sin\theta}{\sin\theta}{\cos\theta}$.
The \emph{Hadamard transform} $H$ is defined as $\frac{1}{\sqrt{2}}\lmatrices{1}{1}{1}{-1}$.
Let $x$ be any input string of length $n$.

\s

(i) Start with the initial quantum state $\qubit{q_0,q_0}\qubit{r_1}$. Using the first register $\qubit{q_0,q_0}$, we run $M$, which is given by Claim \ref{2head-2rfa-forA3}  to check whether $n$ is of the form $3^{2m}$ for a certain $m\in\nat^{+}$. If $M$ rejects $x$, then so does $N$. Otherwise, $M$ enters a unique accepting
state $\qubit{q_{3},q_{acc}}$ after the first tape head reaches $\dollar$, while the second tape head still scans $a$.

(ii) We use the second register $\qubit{r_1}$ and move only  the first tape head. Henceforth, we shall keep the second tape head staying still.

(iii) Whenever the first tape head reads the symbol $a$, $N$ applies $R_{\theta}$ and moves the first tape head to the left.

(iv) After the first tape head reaches $\cent$, apply  $R_{7\pi/18}$ to the current quantum state of the form $\alpha\qubit{r_1}+\beta\qubit{r_2}$. Apply $V_{\cent}$, defined by $V_{\cent}\qubit{r_1}=\qubit{r_{rej}}$ and  $V_{\cent}\qubit{r_2}=\qubit{r_{acc}}$.

(v) Measure the current quantum state. If $(q_3,q_{acc})$ is observed, we accept $x$ with probability $\sin^2(n\theta+7\pi/18)$; otherwise, we reject $x$ with probability $\cos^2(n\theta+7\pi/18)$.

\s

Note that, if $n=3^{2m}$, $n\theta$ is written as
\[
n\theta = 2\pi \left( \sum_{i=0}^{m-1}(-1)^{L(a^{3^{2i}})}9^{m-i+1} + (-1)^{L(a^{n})}\frac{1}{9} + \sum_{i>m}(-1)^{L(a^{3^{2i}})}\frac{1}{9^{i+1-m}} \right).
\]
We want to assert that $x\in L$ iff $N$ accepts $x$ with probability at least $2/3$. Let $\omega_n$ denote $n\theta+\frac{7\pi}{18}\;\mathrm{mod}\; 2\pi$. If $L(a^n)=1$, then we obtain $\frac{\pi}{2}-\frac{\pi}{36} \leq \omega_n\leq \frac{\pi}{2}+\frac{\pi}{36}$. From those bounds, it follows that $\sin^2(n\theta+\frac{7\pi}{18})=\sin^2(\omega_n) \geq \cos^2(\pi/36) >\frac{2}{3}$. Similarly, if $L(a^n)=0$, then we obtain $-\frac{\pi}{36}\leq \omega_n\leq \frac{\pi}{36}$, from which we conclude that $\cos^2(n\theta+\frac{7\pi}{18}) = \cos^2(\omega_n) >\frac{2}{3}$.
\end{proof}

\section{Classical Simulations of 2QFAs}\label{sec:classical-simulation}

We shall establish a close relationship between 2qfa's and 2pfa's by seeking a simulation of the 2qfa's on multi-head 2pfa's. Since Lemma \ref{folke-complex} allows us to deal only with $\real$-amplitudes for 2qfa's, throughout this section, we shall consider 2qfa's that use real amplitudes only.

\subsection{Multi-Head Classical Finite Automata}\label{sec:multi-head-FA}

We wish to present two classical complexity upper bounds of $\twopqfa_{K}$ and $\twocequalqfa_{K}$ for a reasonable choice of amplitude set $K\subseteq \real$. Given such a subset $K$ of $\real$, the notation $\widehat{K}$ refers to the minimal set that contains $K$ and is also closed under \emph{multiplication} and \emph{addition}. In particular, we obtain $\hat{\rational}=\rational$ and $\hat{\real}=\real$.
Let us recall from Section \ref{sec:classical-cutpoint} that the bracketed notation ``$[t(n)\mbox{-}time]$'' indicates a worst-case time bound $t(n)$ of an underlying finite automata and that this notation has yielded two language families $\twoppfa_{K}(k\mbox{-}head)[\polytime]$ and $\twocequalpfa_{K}(k\mbox{-}head)[\polytime]$, which indicate reasonable upper bounds of $\twopqfa$ and $\twocequalqfa$, respectively.

\begin{theorem}\label{twoEQFA-SL}
Let $K$ be any subset of $\real$ with $\{0,1/2,1\}\subseteq K$.
\begin{enumerate}\vs{-1}
  \setlength{\topsep}{-2mm}%
  \setlength{\itemsep}{1mm}%
  \setlength{\parskip}{0cm}%

\item $\twopqfa_{K} \subseteq \mathrm{2PPFA}_{\widehat{K}}(k\mbox{-}head)[\polytime]$ for a certain index $k\geq2$.

\item  $\twocequalqfa_{K} \subseteq \mathrm{2C_{=}PFA}_{\widehat{K}}(k\mbox{-}head)[\polytime]$ for a certain index $k\geq2$.
\end{enumerate}
\end{theorem}

It is important to note that we impose no restriction (such as ``complete halting'' and ``absolutely halting'') on the running time of 2qfa's when  recognizing languages in $\twopqfa_{K}$ as well as in $\twocequalqfa_{K}$ in Theorem \ref{twoEQFA-SL}.

With the help of Lemma \ref{basic-inclusion}(1--3), we obtain the following immediate corollary of Theorem \ref{twoEQFA-SL}.

\begin{corollary}\label{EQFA-Cequalpfa}
There is an index $k\geq2$ such that, for any set $K\subseteq\complex$, $\twoeqfa_{K} \subseteq \mathrm{2C_{=}PFA}_{\widehat{K}}(k\mbox{-}head)[\polytime]\cap \co\mathrm{2C_{=}PFA}_{\widehat{K}}(k\mbox{-}head)[\polytime]$ and
$\twobqfa_{K}\subseteq \twoppfa_{\widehat{K}}(k\mbox{-}head)[\polytime]\cap \co\twoppfa_{\widehat{K}}(k\mbox{-}head)[\polytime]$.
\end{corollary}

From \cite[Lemmas 1--2 \& Theorem 2]{Mac97}, it follows that, for each fixed index $k\in\nat^{+}$,  $\twoppfa_{\rational}(k\mbox{-}head)[\polytime]$ is \emph{properly} contained within the language family $\mathrm{PL}$. In the quantum setting, Nishimura and Yamakami \cite[Section 1]{NY09} earlier noted class inclusions $\twobqfa_{\algebraic}\subseteq \mathrm{PL} \subseteq \mathrm{P}$ as a consequence of a result in \cite{Wat03}. By combining those two results, we instantly obtain another  corollary of Theorem \ref{twoEQFA-SL}.
This corollary gives a limitation of the recognition power of unbounded-error 2qfa's having $\rational$-amplitudes.

\begin{corollary}\label{PL-vs-2PQFA}
$\twopqfa_{\rational} \subsetneqq \mathrm{PL}$ (and thus $\twobqfa_{\rational}\subsetneqq \mathrm{PL}$).
\end{corollary}

\begin{proof}
Theorem \ref{twoEQFA-SL}(1) ensures  the existence of an index $k\geq2$ for which $\twopqfa_{\rational}\subseteq \twoppfa_{\rational}(k\mbox{-}head)[\polytime]$ holds. Since $\twoppfa_{\rational}(k\mbox{-}head)[\polytime]\subsetneqq \mathrm{PL}$ \cite{Mac97}, it instantly follows that $\twopqfa_{\rational}\subsetneqq \mathrm{PL}$.
\end{proof}

\begin{corollary}
$\twopqfa_{\rational} \neq \twopqfa_{\complex}$.
\end{corollary}

\begin{proof}
Since $\mathrm{PL}\subsetneqq \mathrm{REC}$ holds, Corollary \ref{PL-vs-2PQFA} implies that $\twopqfa_{\rational}\subsetneqq \mathrm{REC}$. Since $\twopqfa_{\complex}\nsubseteq \mathrm{REC}$ by Proposition \ref{QFA-non-recursive-NO1}, we can derive a conclusion that $\twopqfa_{\rational} \neq \twopqfa_{\complex}$.
\end{proof}

Theorem \ref{twoEQFA-SL} is a direct consequence of the following technical lemma regarding a classical simulation of 2qfa's on multi-head 2pfa's.

\begin{lemma}\label{2-head-gap-simulation}
Let $K$ be any subset of $\real$. There exists an index $k\geq2$ that satisfies the following. Given a $K$-amplitude 2qfa $M$, there exist two $k$head-2pfa's $N_1$ and $N_2$ such that (i) $N_1$ and $N_2$ have nonnegative $\widehat{K}$-transition probabilities, (ii) $N_1$ and $N_2$ halt in worst-case $n^{O(1)}$ time, and (iii)  it holds that,  for every $x$,  $(p_{N_1,acc}(x)-p_{N_1,rej}(x)) p_{M,acc}(x) = p_{N_2,acc}(x) - p_{N_2,rej}(x)$.
\end{lemma}

Concerning a quantum function $f\in \#\twoqfa_{K}$ generated by an appropriately chosen 2qfa $M$, if we apply Lemma \ref{2-head-gap-simulation} to $M$, then we obtain two appropriate $k$head-2pfa's $N_1$ and $N_2$ satisfying Conditions (i)--(iii) of the lemma. By setting quantum functions $g_1$, $g_2$, $h_1$, and $h_2$ as  $g_1(x)=p_{N_1,acc}(x)$, $g_2(x)=p_{N_1,rej}(x)$, $h_1(x)=p_{N_2,acc}(x)$, and $h_2(x)=p_{N_2,rej}(x)$ for all inputs $x$, it immediately follows that $(g_1(x)-g_2(x))f(x) = h_1(x)-h_2(x)$.
By the definition of $\#\twopfa_{K}(k\mbox{-}head)[\polytime]$ given in Section \ref{sec:classical-cutpoint}, we conclude that $g_1$, $g_2$, $h_1$, and $h_2$ all belong to $\#\twopfa_{\widehat{K}}(k\mbox{-}head)[\polytime]$. This conclusion
immediately  leads to the following corollary concerning the quantum functions in $\#\twoqfa_{K}$.

\begin{corollary}
Let $K\subseteq\real$ with $K\neq\setempty$. There exists an index $k\geq2$ such that, for any function $f\in \#\twoqfa_{K}$, there exist four functions $g_1,g_2,h_1,h_2\in   \#\twopfa_{\widehat{K}}(k\mbox{-}head)[\polytime]$ satisfying $(g_1(x)-g_2(x))f(x) = h_1(x)-h_2(x)$ for every input $x$.
\end{corollary}

Assuming the validity of Lemma \ref{2-head-gap-simulation},
let us prove Theorem \ref{twoEQFA-SL} below.

\vs{-2}
\begin{proofof}{Theorem \ref{twoEQFA-SL}}
Here, we intend to prove only the first containment $\twopqfa_{K}\subseteq \twoppfa_{\widehat{K}}(k\mbox{-}head)[\polytime]$ since the second one is in essence similar to the first one. Take an arbitrary language $L$ in $\twopqfa_{K}$, witnessed by a certain $K$-amplitude 2qfa, say, $M$; that is, for any input $x$, if $x\in L$, then $p_{M,acc}(x)>1/2$, and otherwise $p_{M,rej}(x)\geq1/2$.
By Lemma \ref{2-head-gap-simulation}, there are two appropriate $k$head-2pfa's $N_1$ and $N_2$ that satisfy Conditions (i)--(iii) of the lemma.
Let us define a new $k$head-2pfa $N$ that behaves in the following way.
On input $x$, from the initial inner state $q_0$, $N$ enters another inner state $q_2$ with probability $1/2$, and $q_1$ and $q_3$ with probability $1/4$ each. Starting in $q_1$, $N$ simulates $N_1$ on $x$, whereas, from $q_2$, $N$ simulates $N_2$ and then flips its outcome (i.e., either accepting or rejecting states). From $q_3$, $N$ enters $q_{acc}$ and $q_{rej}$ with equal probability $1/2$.

It follows that $p_{N,acc}(x) = \frac{1}{4}p_{N_1,acc}(x)+ \frac{1}{2}p_{N_2,rej}(x) + \frac{1}{8}$ and $p_{N,rej}(x) = \frac{1}{4}p_{N_1,rej}(x)+ \frac{1}{2}p_{N_2,acc}(x) +\frac{1}{8}$.
From these equalities and by the definition of $N$, we obtain $p_{N,acc}(x) - p_{N,rej}(x) = \frac{1}{4}(p_{N_1,acc}(x)-p_{N_1,rej}(x)) - \frac{1}{2}(p_{N_2,acc}(x)-p_{N_2,rej}(x))$.
Assume that $x\in L$. Since $p_{M,acc}(x)>1/2$, from the equality $(p_{N_1,acc}(x)-p_{N_1,rej}(x)) p_{M,acc}(x) = p_{N_2,acc}(x)-pp_{N_2,rej}(x)$, it follows that $2(p_{N_2,acc}(x)-p_{N_2,rej}(x))<p_{N_1,acc}(x)-p_{N_1,rej}(x)$. Thus, we obtain $p_{N,acc}(x)-p_{N,rej}(x)>0$. Since $N$ halts absolutely, this is equivalent to $p_{N,acc}(x)>1/2$. Similarly, if $x\notin L$, then we obtain $p_{N,rej}(x)\geq 1/2$.
We therefore conclude that $L$ belongs to $\twoppfa_{\widehat{K}}(k\mbox{-}head)[\polytime]$, as requested.
\end{proofof}

At last, we return to  Lemma \ref{2-head-gap-simulation} and present its proof.

\vs{-2}
\begin{proofof}{Lemma \ref{2-head-gap-simulation}}
Let $\Sigma$ be any alphabet and let $M =(Q,\Sigma,\delta,q_0,Q_{acc},Q_{rej})$ be any $\real$-amplitude 2qfa with acceptance probability $p_{M,acc}(x)$ and rejection probability $p_{M,rej}(x)$ on input $x\in\Sigma^*$.

For convenience, assuming that $Q=\{q_0,q_1,\ldots,q_c\}$ with a constant $c\in\nat^{+}$, we define $Q_{acc}=\{q_{i} \mid i\in A\}$ and $Q_{rej} =\{q_{i}\mid i\in R\}$ for certain index sets $A$ and $R$.
For later use, we set $\delta_{+1}(q,\sigma,p,h) =  \delta(q,\sigma,p,h)$ if $\delta(q,\sigma,p,h)>0$, and $0$ otherwise. Similarly, if $\delta(q,\sigma,p,h)<0$, then we set $\delta_{-1}(q,\sigma,p,h)$ as  $-\delta(q,\sigma,p,h)$; otherwise, we set it as $0$. Clearly, it follows that $\delta(q,\sigma,p,h) = \delta_{+1}(q,\sigma,p,h) - \delta_{-1}(q,\sigma,p,h)$.

Hereafter, let $x$ denote an arbitrary input string of length $n$ in $\Sigma^*$.  Let $CONF_n=Q\times [0,n+1]_{\integer}$ and set $N=|CONF_n|$. First, we review how to evaluate the acceptance probability $p_{M,acc}(x)$ of $M$ on $x$.
Recall a transition matrix $U_{\delta}^{(x)}$, which is
an $N\times N$  real matrix induced from $\delta$ on the input $x$.
Note that, for any $(q,\ell),(p,m)\in CONF_n$, the $((q,\ell),(p,m))$-entry of  $U^{(x)}_{\delta}$ matches $\delta(q,x_{\ell},p,m-\ell)$ if $|m-\ell|\leq1$, and $0$ otherwise.

We denote by $P_{non}$ the projection operator onto the space spanned by  non-halting configurations. Let $D_x=U^{(x)}_{\delta}P_{non}$, which precisely describes a single step of $M$ on the input $x$ if its inner state is not a halting state.
It is easy to see that the acceptance probability of $M$ on $x$ at time $k$ equals
\begin{eqnarray*}
\sum_{j\in A} \sum_{\ell\in[0,n+1]_{\integer}} |\bra{q_j,\ell} D_x^k \ket{q_0,0}|^2
&=& \sum_{j\in A} \sum_{\ell\in[0,n+1]_{\integer}} \bra{q_j,\ell}D_x^k\ket{q_0,0}  \bra{q_j,\ell}D_x^k\ket{q_0,0} \\
&=&
\sum_{j \in A} \sum_{\ell\in[0,n+1]_{\integer}} \bra{(q_j,\ell)}\bra{(q_j,\ell)}(D_x^k\otimes D_x^k)\ket{(q_0,0)}\ket{(q_0,0)}.
\end{eqnarray*}
The last vector $y_{ini}=\qubit{q_0,0}\qubit{q_0,0}$ is associated with the initial configuration of $M$ and the vector $y_{acc,j,\ell} = \qubit{(q_{j},\ell)}\qubit{(q_j,\ell)}$ is associated with the accepting state $q_{j}$ in $Q_{acc}$.  Since $(D_x\otimes D_x)^k = D_x^k\otimes D_x^k$ for any index  $k\in\nat$, the total acceptance probability $p_{M,acc}(x)$ of $M$ on $x$ exactly matches
\begin{equation}\label{eqn:p-acc}
p_{M,acc}(x) = \sum_{k=0}^{\infty} \left( \sum_{j\in A} \sum_{\ell\in[0,n+1]_{\integer}} y_{acc,j,\ell}^{T} (D_x\otimes D_x)^k y_{ini}\right) =  y_{acc}^{T} \left(\sum_{k=0}^{\infty}(D_x^k\otimes D_x^k)\right) y_{ini},
\end{equation}
where $y_{acc} = \sum_{j\in A}\sum_{\ell\in[0,n+1]_{\integer}}y_{acc,j,\ell}$.

We are now focused on the operator $\sum_{k=0}^{\infty}(D_x^k\otimes D_x^k)$ in Eq.(\ref{eqn:p-acc}).
First, we express $D_x$ as a difference $ D_x^{+} - D_x^{-}$ of two nonnegative real matrices $D_x^{+}$ and $D^{-}_x$. For this purpose, we define $D_x^{+}[i,j] = D_x[i,j]$ if $D_x[i,j]>0$, and $D_x^{+}[i,j] =0$ otherwise. Similarly, let $D_x^{-}$ be defined as $D_x^{-}[i,j] = - D_x[i,j]$ if $D_x[i,j]<0$, and  $D_x^{-}[i,j] =0$ otherwise.
In addition, we define $\tilde{D}_x = \lmatrices{D_x^{+}}{D_x^{-}}{D_x^{-}}{D_x^{+}}$.
To specify each entry in $\tilde{D}_x$, we conveniently use an index set $Q\times[0,n+1]_{\integer}\times\{\pm1\}$ so that,  for any two indices $(p,m,b),(q,\ell,a)\in Q\times[0,n+1]_{\integer}\times\{\pm1\}$,  the entry $\tilde{D}_x[(p,m,b),(q,\ell,a)]$ equals $\delta_{ab}(q,x_{\ell},p,m-\ell)$ if $|m-\ell|\leq1$, and $0$ otherwise.

At this point, to make our argument readable, we intend to modify the current coordinate system used for $\tilde{D}_x$,
simply by mapping each point $((p,m,b),(q,\ell,a))$ to a new point $((p,q),(m,\ell),(b,a))$ in $Q^2\times[0,n+1]_{\integer}^2\times\{\pm1\}^2$ and by reassigning to each $(((p_1,p_2),(m_1,m_2),(b_1,b_2)),((q_1,q_2),(\ell_1,\ell_2),(a_1,a_2)))$-entry of $\tilde{D}_x\otimes \tilde{D}_x$ the value of   $\delta_{a_1b_1}(q_1,x_{\ell_1},p_1,m_1-\ell_1) \delta_{a_1b_2}(q_2,x_{\ell_2},p_2,m_2-\ell_2)$ if $|m_1-\ell_1|,|m_2-\ell_2|\leq1$; $0$ otherwise.To express this new coordinate system, we write $CONF_*$ for $Q^2\times[0,n+1]_{\integer}^2\times\{\pm1\}^2$.
For simplicity, we regard each element in $CONF_*$ as a new configuration
of $M$.
Hereafter, $\tilde{D}_x$ is treated as a matrix whose index set is $CONF_*$, and thus $\tilde{D}_x\otimes \tilde{D}_x$ is viewed as an $N'\times N'$ nonnegative real matrix, where $N'=|CONF_*|$.

For each index $a\in\{\pm1\}$, we expand $y_{ini}$ to $\tilde{y}_{ini,a} = \qubit{(q_0,q_0),(0,0),(a,a)}$ and $y_{acc}$ to $\tilde{y}_{acc,a} = \sum_{j\in A} \sum_{\ell\in[0,n+1]_{\integer}} \qubit{(q_{acc,j},q_{acc,j}),(\ell,\ell),(a,a)}$.
It then follows from Eq.(\ref{eqn:p-acc}) that
\begin{equation}\label{eqn:expand-p-acc}
p_{M,acc}(x) = \tilde{y}_{acc,+1}^{T}(\sum_{k=0}^{\infty} (\tilde{D}_x\otimes \tilde{D}_x)^k) \tilde{y}_{ini,+1} -  \tilde{y}_{acc,-1}^{T}(\sum_{k=0}^{\infty}(\tilde{D}_x\otimes \tilde{D}_x)^k) \tilde{y}_{ini,-1}.
\end{equation}
In a similar way, we define $\tilde{y}_{rej,+}$ and $\tilde{y}_{rej,-}$ to characterize $p_{M,rej}(x)$.

We need to enumerate all configurations in $CONF_*$ by introducing an appropriate ordering. Let $conf_1= ((q_1,q_2),(\ell_1,\ell_2),(a_1,a_2))$ and $conf_2 = ((p_1,p_2),(m_1,m_2),(b_1,b_2))$ be any two configurations in $CONF_*$. For convenience, we set $h_1=((q_1,q_2),(a_1,a_2))$ and $h_2=((p_1,p_2),(b_1,b_2))$. Here, we assume an appropriate ordering on $Q^2\times\{\pm1\}^2$. Now, we write $conf_2\leq conf_1$ iff (1) $h_2<h_1$, (2) $h_2=h_1$ and $m_1<\ell_1$, or  (3) $h_2=h_1$, $m_1=\ell_1$, and $m_2\leq \ell_2$. Moreover, we write $conf_2<conf_1$ exactly when  $conf_2\leq conf_1$ and $conf_2\neq conf_1$. This relation $<$  forms a linear ordering.
Using this ordering, we enumerate all elements in  $CONF_*$ as $\{i_1,i_2,\ldots,i_{N'}\}$.
In what follows, we intend to identify each element $conf$ in $CONF_*$ with a number $i$ so that $conf$ is the $i$th element in $CINF_*$; with this convention, we slightly abuse the notation by writing $(-1)^{conf}$ to mean $(-1)^{i}$ as long as this expression is clear from the context.

Since the infinite sum $\sum_{k=0}^{\infty}(\tilde{D}_x\otimes \tilde{D}_x)^k$ converges and $\|\tilde{D}_x\otimes \tilde{D}_x\|<1$ holds, it follows that $I-\tilde{D}_x\otimes \tilde{D}_x$ is invertible and that  $(I-\tilde{D}_x\otimes \tilde{D}_x)^{-1} = \sum_{k=0}^{\infty}(\tilde{D}_x\otimes \tilde{D}_x)^k$. From the last equality, it follows that $\tilde{y}_{acc,a}^{T} (I-\tilde{D}_x\otimes \tilde{D}_x)^{-1} \tilde{y}_{ini,a} = \sum_{j\in A}\sum_{\ell\in[0,n+1]_{\integer}} (I-\tilde{D}_x\otimes \tilde{D}_x)^{-1}[\hat{j}_{a,\ell},i_{0,a}]$, where $i_{0,a}$ is $((q_0,q_0),(0,0),(a,a))$ and $\hat{j}_{a,\ell}$ is $((q_{acc,j},q_{acc,j}),(\ell,\ell),(a,a))$ for indices $a\in\{\pm1\}$ and $\ell\in[0,n+1]_{\integer}$.

We therefore establish the equation $p_{M,acc}(x) = \sum_{a\in\{\pm1\}} \sum_{j\in A}\sum_{\ell\in[0,n+1]_{\integer}} (I-\tilde{D}_x\otimes \tilde{D}_x)^{-1}[\hat{j}_{a,\ell},i_{0,a}]$.
By Laplace's formula (\ie $C^{-1}$ equals the adjoint of $C$ divided by $det(C)$), we conclude that
\begin{equation}\label{eqn:det}
p_{M,acc}(x) = \sum_{a\in\{\pm1\}} \sum_{j\in A} \sum_{\ell\in[0,n+1]_{\integer}} \frac{(-1)^{i_{0,a}+\hat{j}_{a,\ell}}det[(I-\tilde{D}_x\otimes \tilde{D}_x)_{i_{0,a},\hat{j}_{a,\ell}}]}{det(I-\tilde{D}_x\otimes \tilde{D}_x)},
\end{equation}
where the notation ``$C_{i,j}$'' expresses a submatrix obtained from matrix $C$ by deleting row $i$ and column $j$.
Hence, it follows from Eq.(\ref{eqn:det}) that $p_{M,acc}(x)>1-\varepsilon$ iff $\sum_{a\in\{\pm1\}} \sum_{j\in A} \sum_{\ell\in[0,n+1]_{\integer}} (-1)^{i_{0,a}+\hat{j}_{a,\ell}}det[(I-\tilde{D}_x\otimes \tilde{D}_x)_{i_{0,a},\hat{j}_{a,\ell}}] > (1-\varepsilon) det(I-\tilde{D}_x\otimes \tilde{D}_x)$.

Next, we state our key lemma. In the lemma, $f_1$ and $f_2$ denote two special functions  defined by $f_1(x)=(\frac{1}{4|Q|^2}2^{-2\ceilings{\log_2(n+2)}})^{8|Q|^2(n+2)^2-1}$ and $f_2(x) = \frac{1}{2|A|}2^{-2\ceilings{\log_2(n+2)}} f_1(x)$ for every $x\in\Sigma^*$.

\begin{lemma}\label{det-simulation}
There exist a number $k\in\nat^{+}$ and a $k$head-2pfa $N_1$ such that $det[I-\tilde{D}_x\otimes\tilde{D}_x] = f_1(x) [p_{N_1,acc}(x)-p_{N_1,rej}(x)]$ for all $x$. Similarly, a certain $k$head-2pfa $N_2$ satisfies that  $\sum_{a\in\{\pm1\}} \sum_{j\in A} \sum_{\ell\in[0,n+1]_{\integer}} (-1)^{i_{0,a}+\hat{j}_{a,\ell}} det[(I-\tilde{D}_x\otimes\tilde{D}_x)_{i_{0,a},\hat{j}_{a,\ell}}] = f_2(x) [p_{N_2,acc}(x)-p_{N_2,rej}(x)]$ for all $x$. Moreover, $f_1$ in the first part can be replaced with $f_2$ by slightly modifying $N_1$ into another machine $N'_1$.
\end{lemma}

Our proof of this lemma is based on a dextrous implementation of the Mahajan-Vinay algorithm \cite{MV97} on multi-head 2pfa's to compute two special  determinants.
For readability, however, the proof of Lemma \ref{det-simulation} is postponed until Section \ref{sec:proof-of-Lemma}.
It is now easy to complete the proof of Lemma \ref{2-head-gap-simulation}.

\vs{-2}
\begin{proofof}{Lemma \ref{2-head-gap-simulation}}
Take $k$head-2pfa's $N'_1$ and $N_2$ given in Lemma \ref{det-simulation}. By Eq.(\ref{eqn:det}), it follows that $p_{M,acc}(x) = (p_{N_2,acc}(x) - p_{N_2,rej}(x))/(p_{N'_1,acc}(x)-p_{N'_1,rej}(x))$, as requested.
\end{proofof}

\subsection{Proof of Lemma \ref{det-simulation}}\label{sec:proof-of-Lemma}

We still need to prove Lemma \ref{det-simulation} for the completion of the proof of Lemma \ref{2-head-gap-simulation}. For this purpose, we must probabilistically ``generate'' the determinants of two real matrices $I-\tilde{D}_{x}\times \tilde{D}_{x}$ and $(I-\tilde{D}_{x}\times \tilde{D}_{x})_{i_{0,a},\hat{j}_{a,\ell}}$.
To carry out this task, we utilize
an elegant GapL-algorithm of
Mahajan and Vinay \cite{MV97}, who demonstrated in the proof of \cite[Theorem 4]{MV97} how to compute the determinant of an integer matrix using ``closed walk (clow).''
We intend to implement their algorithm on $k$head-2pfa's. For our implementation, however, we need to make various changes to the original GapL-algorithm. Such changes are necessary because a target matrix of their algorithm is an integer matrix and is also given as ``input''; on the contrary, in our case, our target matrix is a real matrix and , moreover, we must produce  ``probabilities'' that express the desired determinants of a given integer matrix.  For this purpose, we produce a probabilistic computation tree whose accepting/rejecting computation paths contribute to the calculation of the determinant of the matrix.
We continue using the notation given in Section \ref{sec:multi-head-FA}.
Additionally, we introduce a basic notion of ``clow sequences'' in terms of our transition amplitudes. For convenience, we use the notation $T$ to express either $I-\tilde{D}_x\otimes \tilde{D}_x$ or  $(I-\tilde{D}_x\otimes \tilde{D}_x)_{i_{0,a},\hat{j}_{a,\ell}}$ and we write $\overline{N}$ for the dimension of $T$.
A {\em clow} over $T$ is a sequence $(c_1,c_2,\ldots,c_m)$ of length $m$ ($m\leq \overline{N}^2$) such that (i) each element $c_i$ is taken from the index set $CONF_*$ of $T$ and (ii) $c_1<c_i$ holds for all indices $i\in[2,m]_{\integer}$. The first element $c_1$ is called a \emph{clow head}. The {\em weight} of this clow is $\prod_{i=1}^{m} T[c_i,c_{i+1}]$, where $c_{m+1}=c_1$. A {\em clow sequence} over $A$  is a sequence $C=(C_1,C_2,\ldots,C_k)$ of clows (where $C_i=(c^{(i)}_1,c^{(i)}_2,\ldots,c^{(i)}_{m_i})$ with clow head $c^{(i)}_1$)
with a strictly increasing sequence of clow heads: $c^{(1)}_1<c^{(2)}_1<\cdots<c^{(k)}_1$. The {\em weight} of this clow sequence $C$ is the product of the weights of all the clows in $C$ and is denoted $weight(C)$. The {\em sign} of $C$ is $sgn(C)=(-1)^{\overline{N}^2+k}$.

The aforementioned result of Mahajan and Vinay \cite{MV97} helps us
calculate the determinant of $T$ using clow sequences over $T$.

\begin{lemma}\label{MV97-application}
Let $T$ be either  $I-\tilde{D}_{x}\times \tilde{D}_{x}$ and $(I-\tilde{D}_{x}\times \tilde{D}_{x})_{i_{0,a},\hat{j}_{a,\ell}}$. It then holds that $det(T)= \sum_{C\in CLOW(T)}sgn(C)weight(C) = \sum_{C\in CLOW(T)\wedge sgn(C)=0}weight(C)  - \sum_{C\in CLOW(T)\wedge sgn(C)=1}weight(C)$, where $CLOW(T)$ is the set of all clow sequences of length $\overline{N}^2-2$ over $T$, where $\overline{N}$ is the dimension of $T$.
\end{lemma}

Hereafter, we aim at constructing the desired multi-head 2pfa's $N_1$ and $N_2$ that ``generate'' two probabilities associated with $det[I-\tilde{D}_x\otimes \tilde{D}_x]$ and $\sum_{a\in\{\pm1\}} \sum_{j\in A} \sum_{\ell\in[0,n+1]_{\integer}} (-1)^{i_{0,a}+\hat{j}_{a,\ell}}det[(I-\tilde{D}_x\otimes \tilde{D}_x)_{i_{0,a},\hat{j}_{a,\ell}}]$, respectively, following a series of technical lemmas.

In Lemmas \ref{counter}--\ref{two-conf-compare}, we shall design four subroutines, which can be properly implemented on multi-head 2pfa's. To improve readability, we shall describe those subroutines in an informal procedural manner and their actual implementations on multi-head 2pfa's are left to the avid reader.

We begin with a simple subroutine, implementing an internal counter, say, $Count$, which enters a designated inner state exactly after $4|Q|^2(n+2)^2$ steps elapse.

\begin{lemma}\label{counter}
There is a 2head-2pfa, implementing deterministically an internal counter $Count$, that takes input of length $n$ and enters a special state $q_{done}$ when exactly $4|Q|^2(n+2)^2$ steps elapse.
\end{lemma}

\begin{proof}
Assume that an input of length $n$ is given on an input tape. The first head starts at cell $0$ and moves to the right with idling for $4|Q|^2-1$ steps at each cell until it finishes scanning $\dollar$. When it stops, the head returns to cell $0$. This process takes exactly $4|Q|^2(n+2)$ steps. We repeat the process for $n+2$ times. This repetition can be counted by moving the second head from cell $0$ to cell $n+1$. In the end of the whole process, we enter a designated inner state $q_{done}$ and halt.
\end{proof}

It is also easy to move the desired number of heads to the current position of head $1$.

\begin{lemma}\label{how-to-copy}
Let $t\in\nat^{+}$ and $h\in[0,n+1]_{\integer}$. There is a $(t+1)$-head 2pfa $M$  that deterministically works as follows. On input of length $n$, $M$ starts with head 1 stationed at cell $h$. The machine $M$ moves heads $2\sim t$ (\ie from head $2$ to head $t$) from cell $0$ to cell $h$ and it returns head 1 to cell $h$ and head $t+1$ to cell $0$. For later reference, we call head $t+1$ a \emph{working head}.
\end{lemma}

\begin{proof}
Let $t'=t+1$ for convenience. Assume that head 1 is initially located at cell $h\in[0,n+1]_{\integer}$ and that all the other $t$ heads are stationed at call $0$. The desired $t'$-head 2pfa $M$ behaves in the following manner. We first reset heads $2\sim t'$ to cell $0$ and move them to the right simultaneously for the same number of steps that require us to move head 1 back to cell $0$ from cell $h$. Heads $2\sim t'$ are now positioned at cell $h$. In a similar way, we use head $t+1$ to make head $1$ return to cell $h$. As a result, head $t+1$ comes back to cell $0$.
Finally, we enter a designated inner state $q_{done}$ and halt.
\end{proof}

An important subroutine is to generate all possible configurations in $CONF_*$ with equal probability, which is roughly $1/4|Q|^2(n+2)^2$.

\begin{lemma}\label{equiprob-generation}
There is a 4head-2pfa $M$ that, on any input of length $n$, generates all configurations in $CONF_*$ using heads $1\sim 2$ (i.e., from head $1$ to head $2$) with equal probability  $\frac{1}{4|Q|^2}2^{-2\ceilings{\log_2(n+2)}}$ and halts in worst-case $O(n)$ time. During a run of $M$, we may have rejecting computation paths. In the end, heads $3\sim 4$ must return to cell $0$.
\end{lemma}

\begin{proof}
We shall describe the behavior of the desired 4head-2pfa $M$, which uses heads $1\sim 4$, where head $4$ is particularly used to
keep ``time'' in order to produce equal probability.
In the end, we ``free'' heads $3\sim 4$ by returning them to call $0$.

\ms

1. Start with all heads located at cell $0$. Generate all indices $((q,p),(a,b))$ in $Q^2\times\{\pm1\}^2$  using $M$'s inner states with equal probability $1/4|Q|^2$ without moving any head.

2. As the initial setup, we flip a fair coin $c\in\{0,1\}$ and move head $1$ for $c$ step and head $4$ for one step, both to the right. Set $s=1$. In Stages $3$--$4$, after each stage, we must check if head $4$ reaches $\dollar$. If this is the case, then we return head $4$ to cell $0$ and then advance to Stage 5.

3. Flip a coin $c\in\{0,1\}$. Assuming that head $4-s$ is at cell $0$ and head $s$ is located at cell $i$, we intend to move head $4-s$ to cell $2i+c$. While head $s$ scans non-$\cent$ symbol, we repeat the following procedure (*); however, if head $4$ reaches $\dollar$ before (*) ends, then we reject the input immediately and halt.
\begin{quote}\vs{-2}
(*) Idle head $s$ for 2 steps and move it to the left. Move heads $4-s$ and $4$ for $2$ steps, both to the right.
\end{quote}\vs{-2}

4. Unless head $4-s$ scans $\dollar$ and $c=1$, we move head $4-s$ for $c$ step and head $4$ for one step, both to the right. Otherwise, we reject the input immediately and halt. Update $s$ to be $4-s$ and go back to Stage 3.

5. If $s=1$, then do nothing at this stage. Assuming that $s=3$ and head $1$ is stationed at cell $0$, we move head $3$ to the left and head $1$ to the right simultaneously until head $3$ reaches $\cent$.

6. We repeat Stages 1--5 using heads $2\sim 3$ (instead of heads $1\&3$). When the procedure ends without rejecting the input, we enter a designated inner state $q_{done}$.
\end{proof}

Recall the linear ordering on $CONF_*$ defined in Section \ref{sec:multi-head-FA}. Given two configurations in $CONF_*$, we can determine which one precedes the other according to this linear ordering.

\begin{lemma}\label{two-conf-compare}
There is a 7head-2pfa $M$ that deterministically works as follows. Let $conf_1,conf_2\in CONF_*$. On any input of length $n$, $M$ starts with both a configuration $conf_1$ using heads $1\sim 2$ and a configuration $conf_2$ using heads $3\sim 4$. The machine $M$ enters $q_1$ if $conf_1\geq conf_2$, and $q_2$ if $conf_1< conf_2$. Moreover, $M$ recovers the given two configurations when it enters either $q_1$ or $q_2$.
\end{lemma}

\begin{proof}
The desired 7head-2pfa $M$ is described below. Let $conf_1= ((q_1,q_2),(\ell_1,\ell_2),(a_1,a_2))$ and $conf_2 = ((p_1,p_2),(m_1,m_2),(b_1,b_2))$ be any two given configurations in $CONF_*$.

\ms

1. Firstly, compare $h_1=((q_1,q_2),(a_1,a_2))$ and $h_2=((p_1,p_2),(b_1,b_2))$ without moving heads. If $h_2<h_1$, then enter $q_1$; if $h_2>h_1$, then enter $q_2$. Hereafter, we assume that $h_1=h_2$.

2. Secondly, we want to compare $(\ell_1,\ell_2)$ and $(m_1,m_2)$ in Stages $3$--$6$. Heads $1\sim 2$ are at present stationed at cells $\ell_1$ and $\ell_2$, respectively, and heads $3\sim 4$ are respectively at cells $m_1$ and $m_2$.

3. We intend to check whether $m_1<\ell_1$. Run a machine given by Lemma \ref{how-to-copy} to move heads $5\sim6$ to the positions of heads $1\&3$, respectively, with the help of working head 7.

4. Move heads $5\sim 6$ to the left simultaneously step by step. If head 6 reaches $\cent$ before head 5 does, then return heads $5\sim 6$ back to cell $0$ and enter $q_1$. If head 5 reaches $\cent$ before head 6 does, then return heads $5\sim 6$ to cell $0$ and enter $q_2$.

5. Here, we assume that heads $5\sim 6$ reach cell $0$ at the same time.
Copy the positions of heads $2\&4$ using heads $5\sim 6$, respectively, with the help of working head 7.

6. We intend to check whether $m_2\leq \ell_2$. Move heads 5 and 6 to the left simultaneously step by step. If head 6 reaches $\cent$ before or at the same time head 5 does, then return heads $5\sim 6$ back to cell $0$ and enter $q_1$. If head 5 reaches $\cent$ before head 6 does, then return heads $5\sim 6$ to cell $0$ and enter $q_2$.
\end{proof}

Hereafter, we shall give the desired algorithms generating two values of  $det[I-\tilde{D}_x\otimes\tilde{D}_x]$ and $\sum_{a\in\{\pm1\}}\sum_{j\in A}\sum_{\ell\in[0,n+1]_{\integer}} (-1)^{i_{0,a}+\hat{j}_{a,\ell}} det[(I-\tilde{D}_x\otimes\tilde{D}_x)_{i_{0,a},\hat{j}_{a,\ell}}]$ in a probabilistic manner. We begin with describing the algorithm for $det[I-\tilde{D}_x\otimes\tilde{D}_x]$. We note that
the matrix $I-\tilde{D}_x\otimes\tilde{D}_x$ satisfies the following.
Let $conf_1=((p_1,p_2),(m_1,m_2),(b_1,b_2))$ and $conf_2=((q_1,q_2),(\ell_1,\ell_2),(a_1,a_2))$ in $CONF_*$.
In the case of $conf_1\neq conf_2$,
$(I-\tilde{D}_x\otimes \tilde{D}_x)[conf_1,conf_2]$ equals $-\delta_{a_1b_1}(q_1,x_{\ell_1},p_1,m_1-\ell_1) \delta_{a_2b_2}(q_2,x_{\ell_2},p_2,m_2-\ell_2)$ if $|m_1-\ell_1|,|m_2-\ell_2|\leq 1$, and $0$ otherwise.
In contrast,  $(I-\tilde{D}_x\otimes \tilde{D}_x)[conf_1,conf_1]$ equals $1-\delta_{+1}(q_1,x_{\ell_1},q_1,0) \delta_{+1}(q_2,x_{\ell_2},q_2,0)$.
To clarify the transition probabilities produced by an application of a single move of $M$, we intentionally write $p[conf_2\leftarrow conf_1]$ in place of
$(I-\tilde{D}_x\otimes\tilde{D}_x)[conf_1,conf_2]$.

To produce the probability expressing $det[I-\tilde{D}_x\otimes\tilde{D}_x]$, recall Lemma \ref{MV97-application}, in which $det[I-\tilde{D}_x\otimes\tilde{D}_x]$ is calculated as $\sum_{C\in CLOW(T)\wedge sgn(C)=0} weight(C) - \sum_{C\in CLOW(T)\wedge sgn(C)=1} weight(C)$, where $T=I-\tilde{D}_x\otimes \tilde{D}_x$. To evaluate these two summations, we need to generate each clow sequence $C$ over $T$ with equal probability and produce its weight $weight(C)$ probabilistically. Finally, if $sgn(C)=0$, then we enter accepting states; otherwise, we enter rejecting states.

\ms
\n{\sc Algorithm for} $det[I-\tilde{D}_x\otimes\tilde{D}_x]$:
\s

1. Set $Count$, a counter, given in Lemma \ref{counter} to be $0$.
Run a 4head-2pfa given in Lemma \ref{equiprob-generation} to generate all possible  configurations in $CONF_*$
with equal probability $\frac{1}{4|Q|^2}2^{-2\ceilings{\log_2(n+2)}}$. Call by $conf_0$ a resulted configuration, which corresponds to a clow head, and let  $conf_0 =((q_1,q_2),(\ell_1,\ell_2),(a_1,a_2))$. We set $sign=0$ and define  $conf_1$ to be $conf_0$.

2. Increment the counter by $1$. Here, we want to generate each clow sequence over $T$ with equal probability by executing Stages 2a--2d as long as   $Count$ is less than  $4|Q|^2(n+2)^2$.

2a. Run a 4head-2pfa given in Lemma \ref{equiprob-generation} to generate all configurations $conf_2$  with equal probability $\frac{1}{4|Q|^2}2^{-2\ceilings{\log_2(n+2)}}$ and run a 7head-2pfa given in Lemma \ref{two-conf-compare} to check if $conf_2\geq conf_0$. If not, then
we enter both accepting states and rejecting states with equal probability $1/2$ and halt to eliminate this case. Otherwise, let  $conf_2=((p_1,p_2),(m_1,m_2),(b_1,b_2))$.

2b. Enter two different inner states, say, $q'$ and $q_{done}$ with probabilities $p[conf_2\leftarrow conf_1]$ and $1-p[conf_2\leftarrow conf_1]$, respectively. In inner state $q_{done}$, accept and reject $x$ with equal probability $1/2$ and halt so that this does not contribute to the calculation of the desired determinant. In inner state $q'$, on the contrary, we determine whether  $conf_0\geq conf_2$ by running a 7head-2pfa given in Lemma \ref{two-conf-compare}. Update $sign$ to be $1-sign$ if $conf_1=conf_2$; do nothing otherwise. The last case is needed because $(I-\tilde{D}_x\otimes\tilde{D}_x)[conf_2,conf_1]$ is not positive and it should not be included.

2c. In the case of $conf_2> conf_0$, we
reset $conf_1$ to be $conf_2$. Run a 4head-2pfa given in Lemma \ref{equiprob-generation} to generate all configurations $conf$ with equal probability $\frac{1}{4|Q|^2}2^{-2\ceilings{\log_2(n+2)}}$. If $conf$ is of the form $((q_0,q_0),(0,0),(+1,+1))$, then we should clear this configuration and go to Stage 2; otherwise, we both accept and reject $x$ with equal probability $1/2$ and halt.

2d. In the case of $conf_2\leq conf_0$, run a 4head-2pfa given in Lemma \ref{equiprob-generation} to generate all configurations $conf_3$ with probability $\frac{1}{4|Q|^2}2^{-2\ceilings{\log_2(n+2)}}$.
Using Lemma \ref{two-conf-compare}, we check if $conf_3>conf_0$.
If $conf_3> conf_0$, then we reset $conf_0$ to be this $conf_3$. Reset $sign$ to be $1-sign$ and go back to Stage 2.
Otherwise, accept and reject $x$ with equal probability $1/2$ and halt.

3. After Stage 2, the counter must have hit $4|Q|^2(n+2)^2$ by finishing the generation of each clow sequence. Enter two different inner states, say, $q''$ and $q_{done}$ with probabilities $p[conf_0\leftarrow conf_1]$ and $1- p[conf_0\leftarrow conf_1]$, respectively. In inner state $q_{done}$,  enter accepting states and rejecting states with equal probability and halt. In inner state $q''$, reset $sign$ to be $1-sign$ if $conf_0=conf_1$; do nothing otherwise.

4. If $sign=1$, then accept; otherwise, reject.

\ms

To produce $\sum_{a\in\{\pm1\}} \sum_{j\in A}\sum_{\ell\in[0,n+1]_{\integer}} (-1)^{i_{0,a}+\hat{j}_{a,\ell}} det[(I-\tilde{D}_x\otimes\tilde{D}_x)_{i_{0,a},\hat{j}_{a,\ell}}]$ as a  probability, we first  describe how to produce $det[(I-\tilde{D}_x\otimes\tilde{D}_x)_{i_{0,a},\hat{j}_{a,\ell}}]$ for fixed indices $j\in A$, $\ell\in[0,n+1]_{\integer}$, and $a\in\{\pm1\}$.

\ms
\n{\sc Algorithm for} $det[(I-\tilde{D}_x\otimes\tilde{D}_x)_{i_{0,a},\hat{j}_{a,\ell}}]$:
\s

The desired algorithm executes Stages $1$--$4$ of the algorithm for $det[I-\tilde{D}_x\otimes\tilde{D}_x]$ except for the following points.
Initially, we are given 2 configurations $i_{0,a}$ and $\hat{j}_{a,\ell}$ in $CONF_*$. By Lemma \ref{how-to-copy}, we can freely copy those configurations without changing the original ones.
In Stages 2b and 3 while generating probabilities $p[conf_2\leftarrow conf_1]$ and $p[conf_0\leftarrow conf_1]$, we first check whether $conf_1=i_{0,a}$ or $conf_2=\hat{j}_{a,\ell}$ since $(I-\tilde{D}_x\otimes\tilde{D}_x)_{i_{0,a},\hat{j}_{a,\ell}}$ contains neither $i_{0,a}$-row nor $\hat{j}_{a,\ell}$-column. If so, accept and reject $x$ with equal probability $1/2$ and halt. Otherwise, we follow the original stages.

\ms
\n{\sc Algorithm for} $\sum_{a\in\{\pm1\}} \sum_{j\in A}\sum_{\ell\in[0,n+1]_{\integer}} (-1)^{i_{0,a}+\hat{j}_{a,\ell}} det[(I-\tilde{D}_x\otimes\tilde{D}_x)_{i_{0,a},\hat{j}_{a,\ell}}]$:
\s

1. Generate $(a,j)\in \{\pm1\}\times A$ using inner states with equal probability $1/2|A|$.

2. Following Stages 2--5 of the algorithm given in Lemma \ref{equiprob-generation}, we generate configurations with head stationed at each cell $\ell\in[0,n+1]_{\integer}$ with equal probability $2^{-2\ceilings{\log_2(n+2)}}$.

3. Run the algorithm for $det[(I-\tilde{D}_x\otimes\tilde{D}_x)_{i_{0,a},\hat{j}_{a,\ell}}]$ except that, before Stage 4, we reset $sign$ to be $sign + i_{0,a}+\hat{j}_{a,\ell}\;\;\mathrm{mod}\;2$ to include the extra term of $(-1)^{i_{0,a}+\hat{j}_{a,\ell}}$, where we identify the elements $i_{0,a}$ and $\hat{j}_{a,\ell}$ with their associated numbers as described before.

\ms

It is tedious but not difficult to check whether the above algorithms correctly compute the intended determinants in Lemma \ref{2-head-gap-simulation}. Therefore, we have completed the proof of the lemma.
\end{proofof}

\section{Challenging Questions}

Throughout this paper, we have extensively studied the exotic behaviors of constant-space quantum computation. Because of their simplicity and the continuation of our early study \cite{NY04b,NY09,NY15,VY14,Yam14}, we have modeled such computation using measure-many 2-way quantum finite automata (or 2qfa's), which were first considered in \cite{KW97} as a quantum-mechanical extension of 2-way probabilistic finite automata (or 2pfa's).
In the past two decades since the introduction of quantum finite automata, we have tried to determine the precise power of quantum computation on those devices.

In this paper, we have resolved a few questions regarding (1) relationships among various acceptance criteria of 2qfa's, (2) bounds of the running time required for 2qfa's to recognize languages, (3) non-recursiveness by the choice of (transition) amplitudes of 2qfa's, and (4) efficient classical simulations of 2qfa's.
Nevertheless, there still remain numerous unsolved questions concerning their behaviors and their computational complexity.

For our future study, we wish to raise a few but important questions, which have left open in this paper.

\begin{enumerate}\vs{-1}
  \setlength{\topsep}{-2mm}%
  \setlength{\itemsep}{1mm}%
  \setlength{\parskip}{0cm}%

\item Strengthen Lemma  \ref{basic-inclusion}(1\&3--4),  Corollary \ref{amplitude-reduction}(2), and Theorem \ref{twoEQFA-SL} by proving that  each of the class inclusions stated in them is actually a proper inclusion.

\item In Theorem \ref{twoEQFA-SL} and Lemma \ref{2-head-gap-simulation}, we have not determined the exact value of positive integer $k$. Determine
    the minimal positive integer $k$ that satisfies Theorem \ref{twoEQFA-SL} and Lemma \ref{2-head-gap-simulation}.

\item We have shown in Proposition \ref{QFA-non-recursive-NO2} that $\twobqfa_{\complex}(2\mbox{-}head,\abshalt)$ contains a non-recursive language. Can we reduce $2$ tape heads in $\twobqfa_{\complex}(2\mbox{-}head,\abshalt)$ to a single tape head (namely, $\twobqfa_{\complex}(\abshalt)\nsubseteq \mathrm{REC}$)?

\item Explore more relationships among language families, such as $\twopqfa_{K}$, $\twobqfa_{K}$, $\twocequalqfa_{K}$, and quantum interactive proof systems of Nishimura and Yamakami \cite{NY04b,NY09,NY15}.

\item As noted in Section \ref{sec:quantum-finite-automata}, unbounded-error and exact-error 1qfa's are no more powerful than their classical counterparts. On the contrary, 2qfa's are quite different in power from 2pfa's. Through Section \ref{sec:classical-simulation}, we have tried to characterize $\twopqfa_{K}$ and $\twocequalqfa_{K}$ in terms of classical computation models. Give the precise characterizations of $\twopqfa_{K}$, $\twobqfa_{K}$, and $\twocequalqfa_{K}$ using appropriate classical models.

\item A multi-head model of qfa's has been briefly discussed in Sections \ref{sec:quantum-finite-automata} and \ref{sec:non-recursive} but little is known for this special model except for an early study of \cite{ABF+99}. When we turn our eyes to a classical case, we already know that multi-head 2pfa's with cut points precisely characterize $\mathrm{PL}$ \cite{Mac97}. Does a similar characterization hold also for multi-head 2qfa's?

\item Can Corollary \ref{amplitude-reduction}(1--2) be extended to $\complex$? Prove or disprove that $\twopqfa_{\complex}=\co\twopqfa_{\complex}$. The same question is still open for $\mathrm{PQP}_{\complex}$, which is a polynomial-time counterpart of $\twopqfa_{\complex}$. See \cite{Yam03} for the $\mathrm{PQP}_{\complex}=\co\mathrm{PQP}_{\complex}$ problem.

\item Many constructions of 2qfa's may be boiled down to appropriate manipulations of quantum functions defined by 2qfa's. Lemma \ref{quantum-functions} has briefly discussed properties of those quantum functions. Explore more properties and develop a theory of quantum functions based on 2qfa's.
\end{enumerate}

\let\oldbibliography\thebibliography
\renewcommand{\thebibliography}[1]{%
  \oldbibliography{#1}%
  \setlength{\itemsep}{0pt}%
}
\bibliographystyle{plain}

\begin{thebibliography}{Gur91}
{\small

\bibitem{ADH97}
L. M. Adleman, J. DeMarrais, and M. A. Huang. Quantum computability. SIAM J. Comput. 26 (1997) 1524--1540.


\bibitem{ABF+99}
A. Ambainis, R. Bonner, R. Freivalds, M. Golovkins, and M. Karpinski. Quantum finite multitape automata. In the Proc. of the 26th Conference on Current Trends in Theory and Practice of Informatics (SOFSEM'99), Lecture Notes in Computer Science, Springer,  vol.1725, pp.340--348, 1999.

\bibitem{ABG+06}
A. Ambainis, M. Beaudry, M. Golovkin, A. \c{K}ikusts, M. Mercer, and D. Th\'{e}rien. Algebraic results on quantum automata. Theory of Computing Systems 39 (2006) 165--188.

\bibitem{AF98}
A. Ambainis and R. Freivalds. 1-way quantum finite automata: strengths, weaknesses and generalizations. In the Proc. of the 39th Annual Symposium on Foundations of Computer Science (FOCS'98), pp.332--341, 1998.





\bibitem{BC01}
A. Bertoni, M. Carpentieri. Analogies and differences between quantum and stochastic automata. Theor. Comput. Sci. 262, 69-–81, 2001.

\bibitem{BV97}
E. Bernstein and U. Vazirani. Quantum complexity theory. SIAM J. Comput. {26} (1997) 1411--1473.

\bibitem{Con93}
A. Condon. The complexity of space bounded interactive proof systems. In {Complexity Theory: Current Research} (eds. Ambos-Spies, \etalc), Cambridge University Press, pp.147--189, 1993.

\bibitem{CHPW98}
A. Condon, L. Hellerstein, S. Pottle, and A. Wigderson. On the power of finite automata with both nondeterministic and probabilistic states. {SIAM J. Comput.} {27} (1998) 739--762.

\bibitem{DS90}
C. Dwork and L. Stockmeyer. A time complexity gap for two-way probabilistic finite-state automata. SIAM J. Comput. {19} (1990) 1011--1023.

\bibitem{DS92a}
C. Dwork and L. Stockmeyer. Finite state verifier I: the power of interaction. J. ACM {39} (1992) 800--828.


\bibitem{FK94}
R. Freivalds and M. Karpinski. Lower space bounds for randomized computation. In the  Proc. of the 21st International Colloquium on Automata, Languages and Programming (ICALP'94), Lecture Notes in Computer Science, Springer, vol.820, pp.580--592, 1994.


\bibitem{Hir10}
M. Hirvensalo. Quantum automata with open time evolution. International Journal of Natural Computing 1 (2010) 70--85.

\bibitem{HU79}
J. E. Hopcroft and J. D. Ullman. Introduction to Automata Theory, Language, and Computation. Addison-Wesley, Massachusetts, 1979.

\bibitem{HJ85}
R. A. Horn and C. R. Johnson. Matrix Analysis. Cambridge University Press, Cambridge,  1985.


\bibitem{Kan89}
J. Ka\c{n}eps. Stochasticity of the languages acceptable by two-way finite probabilistic automata. Diskretnaya Matematika 1 (1989) 63--77 (Russian).  Discrete Mathematics and Applications 1 (1991) 405--421 (English)

\bibitem{KF90}
J. Ka\c{n}eps and R. Freivalds. Minimal nontrivial space complexity of probabilistic one-way Turing machines. In the Proc. of the
Mathematical Foundations of Computer Science (MFCS'90), Lecture Notes in Computer Science, Springer, vol.452, pp.355--361, 1990.



\bibitem{KW97}
A. Kondacs and J. Watrous. On the power of quantum finite state automata. In the Proc. of the 38th Annual Symposium on Foundations of Computer Science  (FOCS'97), pp.66--75, 1997.


\bibitem{Mac93}
I. Macarie. Closure properties of stochastic languages. Technical Report No.441, Computer Science Department, University of Rochester, 1993.

\bibitem{Mac97}
I. I. Macarie. Multihead two-way probabilistic finite automata. Thoery Comput.  Sys., 30 (1997) 91--109.

\bibitem{MV97}
M. Mahajan and V. Vinay. Determinant: combinatorics, algorithms, and complexity. Chicago J. Theor. Comput. Sci. vol. 1997, Article no. 1997-5, 1997.

\bibitem{MC00}
C. Moore and J. Crutchfield. Quantum automata and quantum grammar. {Theor. Comput. Sci.} {237} (2000) 275--306.

\bibitem{NH71}
M. Nasu and N. Honda. A context-free language which is not accepted by a probabilistic automaton. Inf. Control 18 (1971) 233--236.

\bibitem{Nay99}
A. Nayak. Optimal lower bounds for quantum automata and random acess codes. In the Proc. of the 40th Annual Symposium on Foundations of Computer Science (FOCS'99), pp.369--376, 1999.

\bibitem{NC00}
M. A. Nielsen and I. L. Chuang. {Quantum Computation and Quantum Information}, Cambridge University Press, Cambridge, 2000.

\bibitem{NY04a}
H. Nishimura and T. Yamakami. Polynomial time quantum computation with advice. {Inform. Process. Lett.} {90} (2004) 195--204.

\bibitem{NY04b}
H. Nishimura and T. Yamakami. An application of quantum finite automata to interactive proof systems (extended abstract). In the {Proc. of the 9th International Conference on Implementation and Application of Automata} (CIAA 2004), Lecture Notes in Computer Science, Springer, vol.3317, pp.225--236, 2004.

\bibitem{NY09}
H. Nishimura and T. Yamakami. An application of quantum finite automata to interactive proof systems. J. Comput. System Sci. {75} (2009) 255--269.
A complete version of the first half part of \cite{NY04b}.

\bibitem{NY15}
H. Nishimura and T. Yamakami. Interactive proofs with quantum finite automata. {Theoret. Comput. Sci.} 568 (2015) 1--18. A complete version of the second half part of \cite{NY04b}.

\bibitem{Rab63}
M. O. Rabin. Probabilistic automata. Inform. Control {6} (1963) 230--244, 1963.

\bibitem{RS59}
M. O. Rabin and D. Scott. Finite automata and their decision problems. IBM J. Res. Dev. {3} (1959) 114--125.



\bibitem{Sto74}
K. B. Stolarsky. Algebraic Numbers and Diophantine
Approximations. Marcel Dekker, 1974.

\bibitem{TYL10}
K. Tadaki, T. Yamakami, and J. C. H. Lin. Theory of one-tape linear-time Turing machines. Theoret. Comput. Sci. 411 (2010) 22--43.
A preliminary version apperaed in the {Proc. 30th SOFSEM Conference on Current Trends in Theory and Practice of Computer Science} (SOFSEM 2004), Lecture Notes in Computer Science, Vol.2932, pp.335-348, Springer, 2004.

\bibitem{Tur68}
P. Turakainenn. Generalized automata and stochastic languages. Proc. of the American Mathematical Society 21 (1969) 303--309.

\bibitem{VY14}
M. Villagra and T. Yamakami. Quantum and reversible verification of proofs using constant memory space. In the Proc. of the 3rd International Conference on the Theory and Practice of Natural Computing (TPNC 2014), Lecture Notes in Computer Science, Vol.8890, pp.144--156.


\bibitem{Wat03}
J. Watrous. On the complexity of simulating space-bounded quantum computations. Computational Complexity 12 (2003) 48--84.

\bibitem{YS10}
A. Yakary{\i}lmaz and A. C. C. Say. Languages recognized by nondeterministic quantum finite automata. Quantum Information and Computation 10 (2010) 747--770.

\bibitem{YS11}
A. Yakary{\i}lmaz and A. C. C. Say. Unbounded-error quantum computation with small space bounds. Inf. Comput. 209 (2011) 873--892.

\bibitem{Yam03}
T. Yamakami. Analysis of quantum functions. Internat. J. Found. Comput. Sci. 14 (2003) 815--852.
A preliminary version appeared in the Proc. of
the 19th IARCS Annual Conference on
Foundations of Software Technology and Theoretical Computer Science (FSTTCS'99), Lecture Notes in Computer Science, Springer, vol.1738, pp.407--419, 1999.

\bibitem{Yam11}
T. Yamakami. Approximate counting for complex-weighted Boolean constraint satisfaction problems. {Inf. Comput.} 219 (2012) 17--38.

\bibitem{Yam14}
T. Yamakami. Constant-space quantum interactive proofs against multiple provers. {Inform. Process. Lett.} 114 (2014) 611--619.

\bibitem{YY99}
T. Yamakami and A. C. Yao. $\mathrm{NQP}_{\complex} = \co\cequalp$. Inform. Process. Lett. 71 (1999) 63--69.

\bibitem{Yao98}
A. C. Yao. Class Note. Unpublished, Princeton University, 1998.

}
\end{thebibliography}

\end{document}